\documentclass[a4paper,USenglish,cleveref, autoref]{lipics-v2019}

\usepackage{amsthm,amsmath,amssymb}

\newcommand{\1}[1]{{\bf 1}}
\newcommand{\2}[1]{{\bf 2}}

\newcommand{\mc}{{\rm mc}}
\newcommand{\cw}{{\rm cw}}
\newcommand{\tw}{{\rm tw}}
\newcommand{\pw}{{\rm pw}}
\newcommand{\tc}{{\rm tc}}
\newcommand{\mw}{{\rm mw}}
\newcommand{\vc}{{\rm vc}}
\newcommand{\CX}{{\mathcal X}}

\newcommand{\revised}[1]{\textcolor{black}{#1}}

\newtheorem{observation}{Observation}

\title{Parameterized Algorithms for Maximum Cut with Connectivity Constraints} 

\author{Hiroshi Eto}
{Kyushu University, Fukuoka, Japan}
{h-eto@econ.kyushu-u.ac.jp}
{}
{}

\author{Tesshu Hanaka}
{Chuo University, Tokyo, Japan}
{hanaka.91t@g.chuo-u.ac.jp}
{}
{}

\author{Yasuaki Kobayashi}
{Kyoto University, Kyoto, Japan}
{kobayashi@iip.ist.i.kyoto-u.ac.jp}
{}
{}

\author{Yusuke Kobayashi}
{Kyoto University, Kyoto, Japan}
{yusuke@kurims.kyoto-u.ac.jp}
{}
{}

\authorrunning{H. Eto, T. Hanaka, Y.Kobayashi, and Y.Kobayashi}

\Copyright{H. Eto, T. Hanaka, Y.Kobayashi, and Y.Kobayashi}



\category{}

\relatedversion{}

\supplement{}

\acknowledgements{We thank Akitoshi Kawamura and Yukiko Yamauchi for giving an opportunity to discuss in  the Open Problem Seminar in Kyushu University. This work is partially supported by JST CREST JPMJCR1401 and JSPS KAKENHI Grant Numbers JP17H01788, 18H06469, JP16K16010, JP17K19960, and JP18H05291.}

\nolinenumbers 

\hideLIPIcs  

\begin{document}

\maketitle

\begin{abstract}
We study two variants of \textsc{Maximum Cut}, which we call \textsc{Connected Maximum Cut} and \textsc{Maximum Minimal Cut}, in this paper.
In these problems, given an unweighted graph, the goal is to compute a maximum cut satisfying some connectivity requirements.
Both problems are known to be NP-complete even on planar graphs whereas \textsc{Maximum Cut} \revised{on planar graphs} is solvable in polynomial time. 
We first show that these problems are NP-complete even on planar bipartite graphs and split graphs.
Then we give parameterized algorithms using graph parameters such as clique-width, tree-width, and twin-cover number.
Finally, we obtain FPT algorithms with respect to the solution size. 

\end{abstract}

\section{Introduction}

\textsc{Maximum Cut} is one of the most fundamental problems in theoretical computer science. Given a graph and an integer $k$, the problem asks for a subset of vertices such that the number of edges having exactly one endpoint in the subset is at least $k$.
This problem was shown to be NP-hard in Karp's seminal work \cite{Karp1972}. To overcome this intractability, a lot of researches have been done from various view points, such as  approximation algorithms \cite{Goemans1995}, fixed-parameter tractability \cite{Raman2007}, and special graph classes \cite{bodlaender2000,Boyaci2017,Diaz2007,Guruswami1999,Hadlock1975,Orlova1972}. 

In this paper, we study two variants of \textsc{Maximum Cut}, called \textsc{Connected Maximum Cut} and \textsc{Maximum Minimal Cut}.
A cut $(S, V\setminus S)$ is {\em connected} if the subgraph of $G$ induced by $S$ is connected.
Given a graph $G=(V,E)$ and an integer $k$, \textsc{Connected Maximum Cut}  is the problem to determine whether there is a connected cut $(S,V\setminus S)$ of size at least $k$ . This problem is defined in \cite{Haglin1991} and known to be NP-complete even on planar graphs \cite{Hajiaghayi2015} whereas \textsc{Maximum Cut} on planar graphs  is solvable in polynomial time \cite{Hadlock1975,Orlova1972}. 

Suppose $G$ is connected. 
We say that a cut $(S, V \setminus S)$ of $G$ is {\em minimal} if there is no another cut of $G$ whose cutset properly contains the cutset of $(S, V \setminus S)$, where the cutset of a cut is the set of edges between different parts. 
We can also define minimal cuts for disconnected graphs (See Section~\ref{sec:preli}).
\textsc{Maximum Minimal Cut} is the following problem: given a graph $G=(V,E)$ and an integer $k$, determine the existence of a minimal cut $(S,V\setminus S)$ of size at least $k$.
This type of problems, finding a maximum minimal (or minimum maximal) solution on graphs such as \textsc{Maximum Minimal Vertex Cover} \cite{Boria2015,Zehavi2017}, \textsc{Maximum Minimal Dominating Set} \cite{Bazgan2018}, \textsc{Maximum Minimal Edge Cover} \cite{Khoshkhah2019}, \textsc{Maximum Minimal Separator} \cite{Hanaka2017}, \textsc{Minimum Maximal matching} \cite{GJ1979,Yannakakis1980}, and \textsc{Minimum Maximal Independent Set} \cite{Demange1999}, has been long studied.

As a well-known fact, a cut $(S,V\setminus S)$ is minimal if and only if both subgraphs induced by $S$ and $V\setminus S$ are connected when the graph is connected \cite{Diestel2012}. Therefore, a minimal cut is regarded as a {\em two-sided} connected cut, while a connected cut is a {\em one-sided} connected cut\footnote{In \cite{Chaourar2017,Chaourar2019,Haglin1991}, the authors used the term ``connected cut'' for {two-sided} connected cut. In this paper, however, we use ``minimal cut'' for {two-sided} connected cut and ``connected cut'' for {one-sided} connected cut for distinction.}.
Haglin and Venkatean \cite{Haglin1991} showed that deciding if the input graph has a two-sided connected cut (i.e., a minimal cut) of size at least $k$ is NP-complete even on triconnected cubic planar graphs. This was shown by the fact that for any two-sided connected cut on a connected planar graph $G$, the cutset corresponds to a cycle on the dual graph of $G$ and vise versa. Hence, the problem is equivalent to the longest cycle problem on planar graphs~\cite{Haglin1991}.
Recently, Chaourar proved that \textsc{Maximum Minimal Cut} can be solved in polynomial time on series parallel graphs and graphs without $K_5\setminus e$ as a minor in \cite{Chaourar2017,Chaourar2019}. 


Even though there are many important applications of  \textsc{Connected Maximum Cut} and \textsc{Maximum Minimal Cut}  such as image segmentation \cite{Vicente2008}, forest planning \cite{Carvajal2013}, and  computing a market splitting for electricity markets \cite{Grimm2019}, the known results are much fewer than those for \textsc{Maximum Cut} due to the difficult nature of simultaneously maximizing its size and handling the connectivity of a cut.

\subsection{Our contribution}
\begin{table}[tbp]
  \centering
    \caption{The summary of the computational complexity of \textsc{Maximum Cut} and its variants.  \textsc{MC},   \textsc{CMC}, and \textsc{MMC} stand for \textsc{Maximum Cut}, \textsc{Connected Maximum Cut}, and \textsc{Maximum Minimal Cut}.}
  \begin{tabular}{|c| c c c |c c c c| c|} \hline
 & \multicolumn{3}{|c|}{Graph Class} &  \multicolumn{4}{|c|}{Parameter} &kernel\cr \hline
 & Split &Bipartite & Planar & $\cw$ & $\tw$ & $\tc$ & $k$ & $k$ \cr \hline \hline
 \textsc{MC} & NP-c & P & P & $n^{O(\cw)}$ & $2^{\tw}$ & $2^{\tc}$  & $1.2418^k$ & $O(k)$\cr 
 &   \cite{bodlaender2000}&[trivial] & \cite{Hadlock1975,Orlova1972} & \cite{Fomin2014} & \cite{bodlaender2000} &\cite{Ganian2015} &\cite{Raman2007} & \cite{Haglin1991,Mahajan1999} \cr \hline
  \textsc{CMC} &  NP-c &NP-c & NP-c & $n^{O(\cw)}$ & $3^{\tw}$& $2^{2^{\tc}+\tc}$ & $9^k$ & No\cr 
& [Th.~\ref{thm:split:connected}]   & [Th.~\ref{thm:bipartite:Connected}] & \cite{Hajiaghayi2015}& [Th.~\ref{thm:cw:dp}] &[Th.~\ref{thm:treewidth_single:Connected}]  &[Th.~\ref{thm:twincover:concut}] & [Th.~\ref{thm:FPTk:Connected}]  &  [Th.~\ref{thm:NoPolyKernel}]  \cr \hline
     \textsc{MMC} & NP-c & NP-c  & NP-c & $n^{O(\cw)}$ & $4^{\tw}$& $2^{\tc}3^{2^{\tc}}$  & $2^{O(k^2)}$ &No \cr 
   & [Th.~\ref{thm:split:maxmin}] & [Th.~\ref{thm:bipartite:maxmin}]  &\cite{Haglin1991}& [Th.~\ref{thm:cw:dp}]& [Th.~\ref{thm:treewidth_single:Maxmin}] &[Th.~\ref{thm:twincover:maxmin}]  & [Th.~\ref{FPT:solution}]  & [Th.~\ref{thm:NoPolyKernel}] \cr \hline
  \end{tabular}
  \label{table:contribution}
\end{table}

Our contribution is summarized in Table~\ref{table:contribution}.
We prove that both \textsc{Connected Maximum Cut} and \textsc{Maximum Minimal Cut} are NP-complete even on planar bipartite graphs.
Interestingly, although \textsc{Maximum Cut} can be solved in polynomial time on planar graphs \cite{Hadlock1975,Orlova1972} and bipartite graphs, both problems are intractable even on the intersection of these tractable classes. 
We also show that the problems are NP-complete on split graphs.

To tackle to this difficulty, we study both problems from the perspective of the parameterized complexity.
We give $O^*({\tw}^{O(\tw)})$-time algorithms for both problems\footnote{The $O^*(\cdot)$ notation suppresses polynomial factors in the input size.}, where $\tw$ is the tree-width of the input graph.
Moreover, we can improve the running time using the rank-based approach \cite{Bodlaender2015} to $O^*(c^{\tw})$ for some constant $c$ and using the Cut \& Count technique \cite{Cygan2011} to $O^*(3^{{\tw}})$ for \textsc{Connected Maximum Cut} and $O^*(4^{\tw})$ for \textsc{Maximum Minimal cut} with randomization.
Let us note that our result generalizes the polynomial time algorithms for \textsc{Maximum Minimal Cut} on series parallel graphs and graphs without $K_5\setminus e$ as a minor due to Chaourar~\cite{Chaourar2017,Chaourar2019} since such graphs are tree-width bounded~\cite{Robertson1986}.

Based on these algorithms, we give $O^*(2^{k^{O(1)}})$-time algorithms for both problems.
For \textsc{Connected Maximum Cut}, we also give a randomized $O^*(9^k)$-time algorithm.
As for polynomial kernelization, we can observe that \textsc{Connected Maximum Cut} and  \textsc{Maximum Minimal Cut}  admit no polynomial kernel when parameterized by solution size $k$ under a reasonable complexity assumption (see, Theorem~\ref{thm:NoPolyKernel}).

We also consider different structural graph parameters.
We design XP-algorithms for both problems when parameterized by clique-width $\cw$.
Also, we give $O^*(2^{2^{\tc}+\tc})$-time and $O^*(2^{\tc}3^{2^{\tc}})$-time FPT algorithms for \textsc{Connected Maximum Cut} and \textsc{Maximum Minimal Cut}, respectively, where $\tc$ is the minimum size of a twin-cover of the input graph.

\subsection{Related work}
\textsc{Maximum Cut} is a classical graph optimization problem and there are many applications in practice. The problem is known to be NP-complete even on split graphs, tripartite graphs, co-bipartite graphs, undirected path
graphs \cite{bodlaender2000}, unit disc graphs \cite{Diaz2007}, and total graphs \cite{Guruswami1999}.
On the other hand, it is solvable in polynomial time on bipartite graphs, planar graphs \cite{Hadlock1975,Orlova1972}, line graphs \cite{Guruswami1999}, and proper interval graphs \cite{Boyaci2017}.
For the optimization version of \textsc{Maximum Cut}, there is a $0.878$-approximation algorithm using semidefinite programming \cite{Goemans1995}. 
As for parameterized complexity, \textsc{Maximum Cut} is FPT \cite{Raman2007} and has a linear kernel \cite{Haglin1991,Mahajan1999} when parameterized by the solution size $k$. Moreover, the problem is FPT when parameterized by tree-width \cite{bodlaender2000} and twin-cover number \cite{Ganian2015}. Fomin et al. \cite{Fomin2014} proved that \textsc{Maximum Cut} is W[1]-hard but XP when parameterized by clique-width.

\textsc{Connected Maximum Cut} was proposed in \cite{Haglin1991}. The problem is a connected variant of \textsc{Maximum Cut} as with \textsc{Connected Dominating Set} \cite{Guha1998} and \textsc{Connected Vertex Cover} \cite{Cygan2012}.
Hajiaghayi et al. \cite{Hajiaghayi2015} showed that the problem is NP-complete even on planar graphs whereas it is solvable in polynomial time on bounded treewidth graphs.
For the optimization version of \textsc{Connected Maximum Cut}, they proposed a polynomial time approximation scheme (PTAS) on planar graphs and more generally on bounded genus graphs and an $\Omega(1/\log n)$-approximation algorithm on general graphs. 

\textsc{Maximum Minimal Cut} was considered in \cite{Haglin1991} and shown to be NP-complete on planar graphs. Recently, Chaourar proved that the problem can be solved in polynomial time on series parallel graphs \cite{Chaourar2017} and graphs without $K_5\setminus e$ as a minor \cite{Chaourar2019}. 

As another related problem, \textsc{Multi-Node Hubs} was proposed by Saurabh and Zehavi \cite{Saurabh2018}: Given a graph $G=(V,E)$ and two integers $k, p$, determine whether there is a connected cut of size at least $k$ such that the size of the connected part is {\em exactly} $p$.
They proved that \textsc{Multi-Node Hubs} is W[1]-hard with respect to $p$, but solvable in time  $O^*(2^{2^{O(k)}})$.
As an immediate corollary of their result, we can solve \textsc{Connected Maximum Cut} in time  $O^*(2^{2^{O(k)}})$ by solving \textsc{Multi-Node Hubs} for each $0\le p\le n$.
\begin{proposition}[\cite{Saurabh2018}]\label{prop:MNH}
\textsc{Connected Maximum Cut} can be solved in time  $O^*(2^{2^{O(k)}})$.
\end{proposition}
In this paper, we improve the running time in Proposition~\ref{prop:MNH} by giving an $O^*(2^{O(k)})$-time algorithm for \textsc{Connected Maximum Cut} in Section~\ref{sec:FPT:solutionsize}.

\section{Preliminaries}\label{sec:preli}
In this paper, we use the standard graph notations.
Let $G = (V, E)$ be an undirected graph.
For $V'\subseteq V$, we denote by $G[V']$ the subgraph of $G$ induced by $V'$. 
We denote the open neighbourhood of $v$ by $N(v)$ and the closed neighbourhood by $N[v]$.

A {\em cut} of $G$ is a pair $(S, V \setminus S)$ for some subset $S \subseteq V$. Note that we allow $S$ (and $V \setminus S$) to be empty. For simplicity, we sometimes denote a cut $(S, V \setminus S)$ by $(S_1, S_2)$ where $S_1=S$ and $S_2=V\setminus S$.
If the second part $V \setminus S$ of a cut is clear from the context, we may simply denote $(S, V\setminus S)$ by $S$.
The {\em cutset} of $S$, denoted by  $\delta(S)$, is the set of {\em cut edges} between $S$ and $V\setminus S$. The size of a cut is defined as the number of edges in its cutset (i.e., $|\delta(S)|$). A cut $S$ is {\em connected} if the subgraph induced by $S$ is connected. 
We say that a cutset is {\em minimal} if there is no non-empty proper cutset of it. 
A cut is {\em minimal} if its cutset is minimal.
It is well known that for every  minimal cut $S$, $G[S]$ and $G[V\setminus S]$ are connected when $G$ is connected~\cite{Diestel2012}.
If $G$ has two or more connected component, every cutset of a minimal cut of $G$ corresponds to a minimal cutset of its connected component.
Therefore, throughout the paper, except in Theorem~\ref{thm:NoPolyKernel}, we assume the input graph $G$ is connected.
Let $p$ be a predicate.
We define the function $[p]$ as follows: if $p$ is true, then $[p]=1$, otherwise $[p]=0$. 

\subsection{Graph parameters}
In this paper, we use several graph parameters such as tree-width, clique-width, and twin-cover number. 

\paragraph*{Tree-width}
\begin{definition}
A {\em tree decomposition} of a graph $G=(V,E)$ is defined as 
a pair $\langle {\cal X}, T\rangle$, where $T$ is a tree with node set $I$ and ${\cal
X}=\{X_i \mid i\in I \}$ is a collection of subsets, called {\em bags}, of $V$ such that:
\begin{itemize}
\item $\bigcup_{i\in I} X_i =V$;
\item For every $\{u, v\}\in E$, there exists an $i\in I$ such that 
	   $\{u, v\} \subseteq X_i$;
\item For every $i, j, k \in I$, if $j$ lies on the
	   path between $i$ and $k$ in $T$, then $X_i \cap X_k \subseteq X_j$.   
\end{itemize}
The {\em width} of a tree decomposition is defined as $\max_{i\in I} |X_i| - 1$ and the {\em tree-width} of $G$, denoted by $\tw(G)$,  is defined as  the minimum width among all possible tree decompositions of $G$.

Moreover, if $T$ of a tree decomposition is a path, it is called a {\em path decomposition} and the {\em path-width} of $G$, denoted by $\pw(G)$,  is defined as  the minimum width among all possible path decompositions of $G$.
\end{definition}

\begin{definition}\label{def:nicetree}
 A tree decomposition  $\langle {\cal X}, T\rangle$ is called a {\em nice tree decomposition} if it satisfies the following:
\begin{description}
\item[1.] $T$ is rooted at a designated node $r(T) \in I$ satisfying $X_{r(T)}=\emptyset$, called 
the {\em root node}.
\item[2.] Every node of the tree $T$ has at most two children.
\item[3.] Each node $i$ in $T$ has one of the following five types:
\begin{itemize}
\item A {\em leaf} node $i$ has no children and its bag $X_i$ satisfies $X_i = \emptyset$,
\item An {\em introduce vertex} node $i$  has exactly one child $j$ with $X_i = X_j \cup \{v\} $ for a vertex $v\in V$,
\item An {\em introduce edge} node $i$ has exactly one child $j$ and  labeled with an edge $\{u, v\} \in E$ where $u, v \in X_i$ and  $X_i = X_j$,
\item A {\em forget} node $i$ has exactly one child $j$ and satisfies $X_i = X_j \setminus \{v\}$ for
a vertex $v \in V$. and
\item A {\em join} node $i$ has exactly two children $j_1, j_2$ and satisfies $X_{j_1}=X_i$ and $X_{j_2}=X_i$.
\end{itemize} 
\end{description}
\end{definition}
We can transform any tree decomposition to a nice tree decomposition with $O(n)$ bags
and the same width in linear time~\cite{Cygan2015}. 
For each node $i$ in a nice tree decomposition $T$, we define a subgraph $G_i=(V_i, E_i)$, where $V_i$ is the union of all bags $X_j$ such that $j=i$ or $j$ is a descendant of $i$ and $E_i\subseteq E$ is the set of all edges introduced at $i$ (if $i$ is an introduce edge node) or a descendant of $i$.

\paragraph*{Clique-width}
\begin{definition}\label{def:clique-width}
Let $w$ be a positive integer.
A {\em $w$-graph} is a graph whose vertices labeled by an integer in $\{1, 2, \ldots, t\}$.
The {\em clique-width} of $G$, denoted by $\cw(G)$, is the minimum integer $w$ such that $G$ can be constructed by means of repeated application of the following operations.
\begin{itemize}
    \item Add a new vertex with label $i \in \{1, 2, \ldots, t\}$;
    \item Take a disjoint union of $w$-graphs $G_1$ and $G_2$;
    \item Take two labels $i$ and $j$ and add an edge between every pair of vertices labeled by $i$ and by $j$; 
    \item Relabel the vertices of label $i$ to label $j$.
\end{itemize}
\end{definition}
\paragraph*{Twin-cover}
Two vertices $u,v$ are called {\em twins} if both $u$ and $v$ have the same neighbors. 
Moreover, if twins $u,v$ have edge $\{u,v\}$, they are called {\em true twins} and the edge is called a {\em twin edge}. 
Then a {\em twin-cover} of $G$ is defined as follows.
\begin{definition}[\cite{Ganian2015}]\label{def:twincover}
A set of vertices $X$ is {\em twin-cover} of $G$ if every edge $\{u,v\}\in E$ satisfies either
\begin{itemize}
\item $u\in X$ or $v\in X$, or
\item $u,v$ are true twins.
\end{itemize}  
The {\em twin-cover number} of $G$, denoted by $\tc(G)$, is defined as the size of minimum twin-cover in $G$. 
\end{definition}
An important observation is that the complement of a twin-cover $X$ induces disjoint cliques.
Moreover, for each clique $Z$ of $G[V \setminus X]$, $N(u) \cap X = N(v) \cap X$ for every $u, v \in Z$~\cite{Ganian2015}.

A {\em vertex cover} $X$ is the set of vertices such that for every edge, at least one endpoint is in $X$.
The {\em vertex cover number} of $G$, denoted by $\vc(G)$, is defined as the size of minimum vertex cover in $G$. Since every vertex cover of $G$ is also a twin-cover of $G$, $\tc(G)\le \vc(G)$ holds.

For clique-width, tree-width, path-width, twin-cover number, and vertex cover number, the following relations hold.
\begin{proposition}[\cite{Bodlaender1995b,Courcelle2000,Ganian2015}]\label{prop:parameters}
For any graph $G$, the following inequalities hold: $\cw(G)\le 2^{\tw(G)+1}+1$, $\cw(G)\le \pw(G)+1$, $\tw(G)\le \pw(G)\le \vc(G)$, $\cw(G)\le 2^{\tc(G)}+\tc(G)$, and $\tc(G)\le \vc(G)$.
\end{proposition}
From Proposition \ref{prop:parameters}, we can illustrate the parameterized complexity of \textsc{Maximum Cut}, \textsc{Connected Maximum Cut}, and \textsc{Maximum Minimal Cut} associated with graph parameters in Figure~\ref{structural_parameter}.

\begin{figure}[t]
    \begin{center}
   \includegraphics[clip,width=5cm]{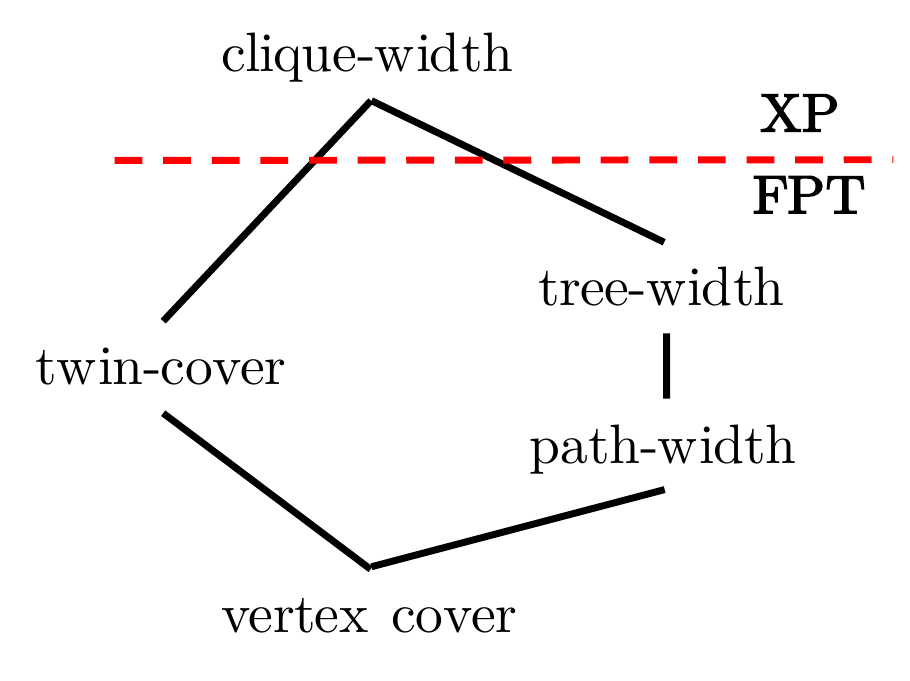}
    \end{center}
    \caption{Graph parameters and the parameterized complexity of \textsc{Maximum Cut}, \textsc{Connected Maximum Cut}, and \textsc{Maximum Minimal Cut}. Connections between two parameters imply the above one is bounded by some function in the below one.}
    \label{structural_parameter}
\end{figure}

\section{Computational Complexity on Graph Classes}\label{sec:graphclass}
In this section, we prove that \textsc{Connected Maximum Cut} and \textsc{Maximum Minimal Cut} are NP-complete on planar bipartite graphs and split graphs.

\subsection{Planar bipartite graphs}\label{sec:bipartite}

\begin{figure}[tbp]
  \begin{center}
       \includegraphics[width=80mm]{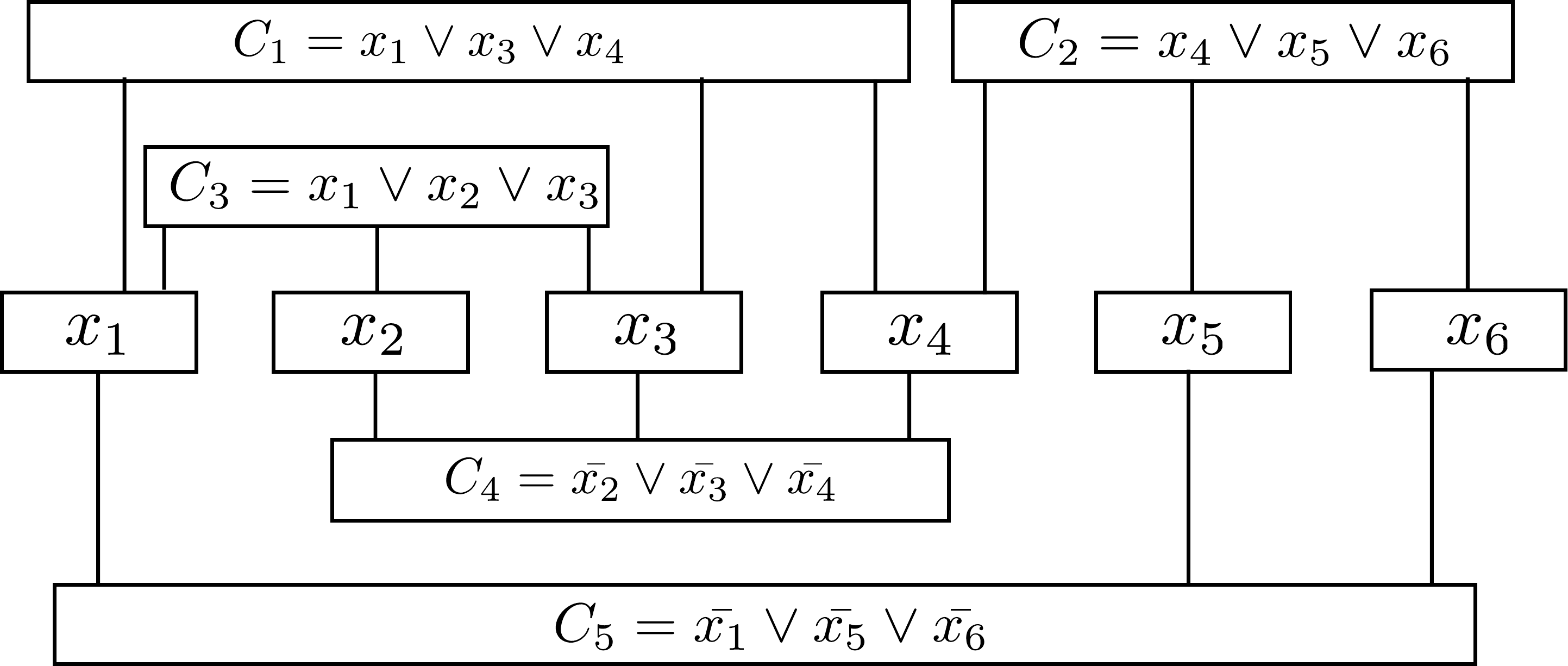}
  \end{center}
  \begin{center}
  \subcaption{A monotone rectilinear representation of an instance of {\sc Planar Monotone 3-SAT}.}
  \end{center}
  \begin{center}
   \includegraphics[width=130mm]{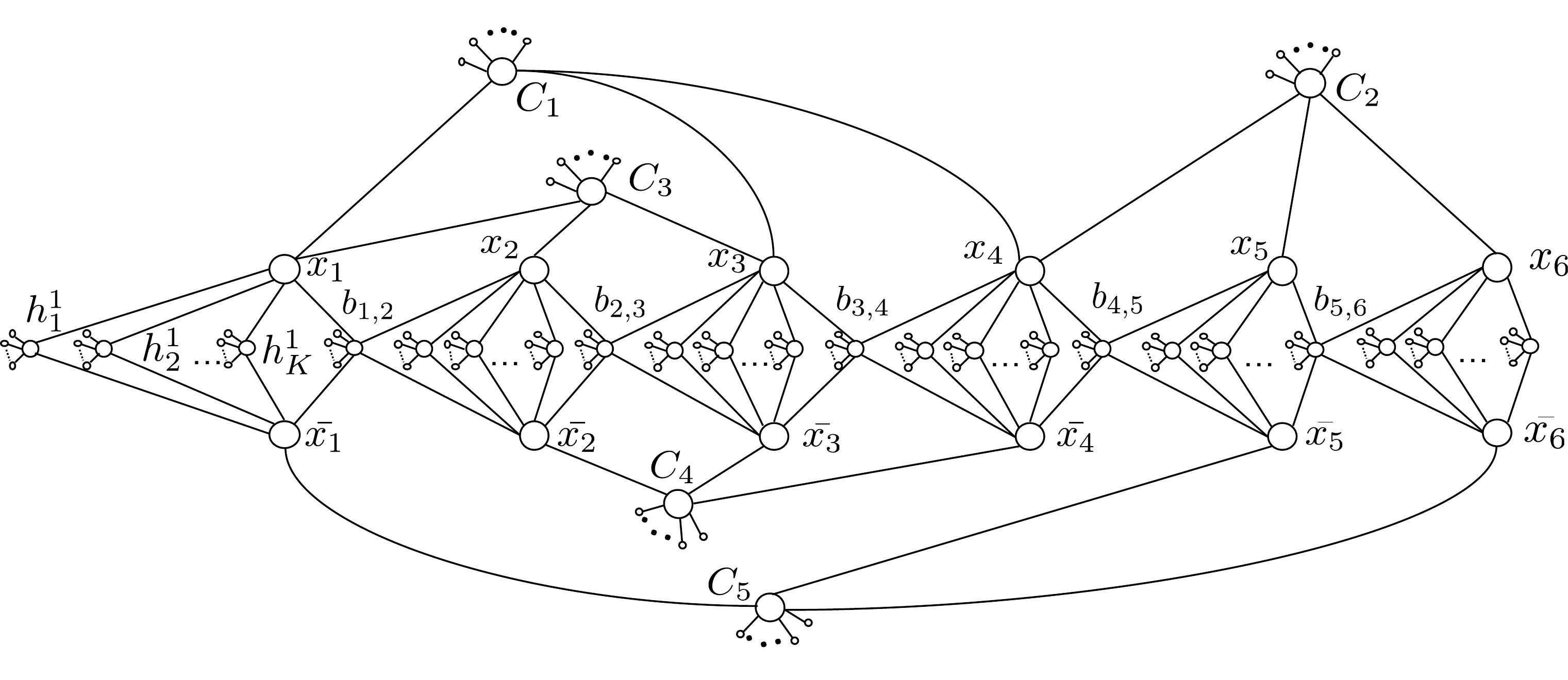}
  \end{center}
  \begin{center}
  \subcaption{The reduced graph $H_\phi$. }
  \end{center}
  \caption{An example illustrating (a) the rectilinear representation of an formula $\phi=(x_1\lor x_3 \lor x_4)\land (x_4\lor x_5 \lor x_6)\land (x_1\lor x_2 \lor x_3)\land (\bar{x_2}\lor \bar{x_3} \lor \bar{x_4})\land(\bar{x_1}\lor \bar{x_5} \lor \bar{x_6})$ of \textsc{Planar Monotone 3-SAT} and (b) the reduced graph $H_\phi$.}
 \label{PM3SATtoCMConPB}
\end{figure}

\begin{theorem}\label{thm:bipartite:Connected}
\textsc{Connected Maximum Cut} is NP-complete on planar bipartite graphs.
\end{theorem}
\begin{proof}
We give a polynomial time reduction for \textsc{Connected Maximum Cut} on planar bipartite graphs. 
The reduction is based on the proof of Theorem $4$ in \cite{Hajiaghayi2015}, which proves that \textsc{Connected Maximum Cut} is NP-hard on planar graphs.
We transform the planar reduced graph in \cite{Hajiaghayi2015} into our planar bipartite graph by using additional vertices, called {\em bridge vertices}.

In this proof, we reduce an instance of \textsc{Planar Monotone 3-SAT}, which is known to be NP-complete \cite{Berg2009}, to a planar bipartite instance of \textsc{Connected Maximum Cut}.
An instance of \textsc{Planar Monotone 3-SAT} consists of a 3-CNF formula $\phi$ satisfies the following properties: (a) each clause contains either all possible literals or all negative literals, (b) the associated bipartite graph $G_{\phi}$ is planar, and (c) $G_{\phi}$ has a {\em monotone rectilinear representation}.
In a monotone rectilinear representation of $G_{\phi}$, the variable vertices are drawn on a straight line in the order of their indices and each positive (resp., negative) clause vertex is drawn in the ``positive side'' (resp., ``negative side'') of the plane defined by the straight line (See Figure~\ref{PM3SATtoCMConPB}).

\noindent\textbf{The Reduction.}
Given a formula $\phi$  of \textsc{Planar Monotone 3-SAT} with a monotone rectilinear representation as in Figure~\ref{PM3SATtoCMConPB} (a), let $X=\{x_1, x_2,\dots, x_n\}$ be a set of variables and ${\mathcal C}=\{C_1, C_2, \dots, C_m\}$ be a set of clauses.
Let $K>m^2$ be sufficiently large. 
Then we create the graph $H_{\phi}=(V,E)$ as follows (see Fig.\ref{PM3SATtoCMConPB}). 
For each variable $x_i\in X$, we create two {\em literal vertices} $v(x_i)$ and $v(\bar{x}_i)$ corresponding to the literals $x_i$ and $\bar{x}_i$, respectively. 
Moreover, we add $K$ {\em helper vertices} $h^i_1, \ldots, h^i_{K}$ and connect $h^i_{k}$ to $v(x_i)$ and $v(\bar{x}_i)$ for each $k=1, \ldots, K$.
For every clause $C_j\in {\mathcal C}$, we create a {\em clause vertex} $v(C_j)$ and connect $v(x_i)$ (resp., $v(\bar{x}_i)$) to $v(C_j)$ if $C_j$ contains $x_i$ (resp., $\bar{x}_i$).
Moreover, we attach $\sqrt{K}$ pendant vertices to each $v(C_j)$.

Then we attach $K$ pendant vertices to each helper vertex  $h^i_k$. 
Finally, we add a {\em bridge vertex} $b_{i,i+1}$ with  $K$ pendant vertices that we make adjacent to each  $v(x_{i}), v(\bar{x}_{i}), v(x_{i+1})$, and $v(\bar{x}_{i+1})$ for $1\le i\le n-1$. 
We denote by $H_\phi$ the graph we obtained.
Notice that we can draw $H_\phi$ in the plane according to a monotone rectilinear representation.
Moreover, $H_\phi$ is bipartite since we only add helper and bridge vertices, which have a neighbor only in literal vertices, and pendant vertices to the planar drawing of $G_\phi$.

Clearly, this reduction can be done in polynomial time.
To complete the proof, we prove the following claim.
\begin{claim*}
A formula $\phi$ is satisfiable if and only if there is a connected maximum cut of size at least $m\sqrt{K} + nK^2 + (2n-1)K +2(n-1)$ in $H_\phi$. 
\end{claim*}
\begin{proof}
Let $V_X$, $V_C$, $V_H$, $V_B$, and $V_P$ be the set of literal vertices, clause vertices, helper vertices, bridge vertices, and pendant vertices, respectively.

\noindent ($\Rightarrow$) We are given a satisfiable assignment $\alpha$ for $\phi$. 
For $\alpha$, we denote a true literal by $l_i$. We also call $v(l_i)$ a {\em true literal vertex}.
Let $S=\bigcup_{i=1}^n \{v(l_i)\}\cup V_C\cup V_H\cup V_B$.
That is, $S$ consists of  the set of true literal vertices, all the clauses vertices, all helper vertices and bridge vertices.
Observe that the induced subgraph by $S$ is connected.
This follows from the facts that each clause has at least one true literal and literal vertices are connected by bridge vertices.

Finally, we show that $|\delta(S)|\ge m\sqrt{K} + nK^2 + (2n-1)K +2(n-1)$.
Since each clause vertex has $\sqrt{K}$  pendant vertices and each helper vertex $K$ pendant vertices, there are $m\sqrt{K} + nK^2$ cut edges.
Moreover, each bridge vertex has $K$ cut edges incident to its pendant vertices and two cut edges incident to literal vertices not in $S$.
Finally, since either $v(x_i)$ or $v(\bar{x_i})$ is not in $S$, there are $nK$ cut edges between literal vertices and helper vertices.
Therefore, we have $|\delta(S)|\ge m\sqrt{K} + nK^2 + (n-1)(K+2) + nK = m\sqrt{K} + nK^2 + (2n-1)K + 2(n-1)$.

\medskip
\noindent ($\Leftarrow$)
We are given a connected cut $S$ in $H_\phi$ such that $|\delta(S)|\ge m\sqrt{K} + nK^2 + (2n-1)K +2(n-1)$.
Here, we assume without loss of generality that $S$ is an optimal connected cut of $H_\phi$.
Suppose, for contradiction, that neither of $v(x_i)$ and $v(\bar{x}_i)$ is contained in $S$ for some variable $x_i$. 
Then, all helper vertices $h^i_k$ cannot be contained in $S$ due to the connectivity of $S$. 
There are $m\sqrt{K} + 3m + 2nK + (K + 4)(n-1)$ edges except for those between helper vertices and its pendant vertices.
It immediately follows that $|\delta(S)| \le m\sqrt{K} + 3m + 2nK + (K + 4)(n-1) + nK^2$. 
Since $K>m^2$ is sufficient large, this contradicts the assumption that $|\delta(S)|\ge m\sqrt{K} + nK^2 + (2n-1)K +2(n-1)$.
Thus, at least one literal vertex must be contained in $S$ for each $x_i$.

Next, we show that every helper vertex must be contained in $S$.
Suppose a helper vertex $h^i_k$ is not contained in $S$.
Then, all $K$ pendant vertices attached to $h^i_k$ is not contained in $S$ due to the connectivity of $S$.
Since at least one literal vertex of $x_i$ is contained in $S$, we can increase the size of the cut by moving $h^i_k$ to $S$, contradicting the optimality of $S$.
This contradicts the optimality of $S$.
Therefore, we assume that every helper vertex is contained in $S$.
Similar to helper vertices, we can prove that every bridge vertex is contained in $S$.

Then, we observe that exactly one literal vertex must be contained in $S$ for each $x_i$.
Suppose that both $v(x_i)$ and $v(\bar{x}_i)$ are contained in  $S$ for some $x_i$.
Since all helper vertices and bridge vertices are contained in $S$, we may increase the size of the cut by moving either of $v(x_i)$ and $v(\bar{x}_i)$ to $V \setminus S$.
However, there are some issues we have to consider carefully.
Suppose that $v(x_i)$ is moved to $V \setminus S$.
Then, some clause vertices $v(C_j)$ in $S$ can be disconnected in $G[S]$.
If so, we also move $v(C_j)$ together with its pendant vertices to $V \setminus S$.
Since there are at least $K + 1$ cut edges newly introduced but at most $m(\sqrt{K} + 3)$ edges removed from the cutset, the size of the cutset is eventually increased, also contradicting the optimality of $S$.

Finally, we show that every clause vertex is, in fact, contained in $S$.
Suppose that $|S \cap V_C| = m' < m$.
If $v(C_j)$ is not in $S$, its pendant vertices are also not in $S$.
Due to the optimality of $S$, the pendant vertices of every helper vertex and every bridge vertex is in $V \setminus S$.
Thus, we have $|\delta(S) \cap \delta(V_H \cup V_B)| = nK(K+1) + (K+2)(n-1) = nK^2 + (2n-1)K+2(n-1)$.
Therefore, $|\delta(S)| = |\delta(S) \cap \delta(V_C)| +  |\delta(S) \cap \delta(V_H \cup V_B)| \le m'\sqrt{K} +3(m-m')+ nK^2 + (2n-1)K +2(n-1) < m\sqrt{K} + nK^2 + (2n-1)K +2(n-1)$ as $K > m^2$.
This is also contradicting to the assumption of the size of the cut.

To summarize, exactly one literal vertex is in $S$ for each variable and every clause vertex is in $S$.
Since $G[S]$ is connected, every clause vertex is adjacent to a literal vertex included in $S$.
Given this, we can obtain a satisfying assignment for $\phi$.
\end{proof}
This completes the proof.
\end{proof}

\begin{theorem}\label{thm:bipartite:maxmin}
\textsc{Maximum Minimal Cut} is NP-complete on planar bipartite subcubic graphs.
\end{theorem}
\begin{proof}
We give a reduction from \textsc{Maximum Minimal Cut}  on planar cubic graphs, which is known to be NP-complete \cite{Haglin1991}.
Given a connected planar cubic graph $G=(V,E)$, we split each edge $e=\{u,w\}\in E$ by a vertex $v_e$, that is, we introduce a new vertex $v_e$ and replace $e$ by $\{u,v_e\}$ and $\{w, v_e\}$.
Let $V_E=\{v_e\mid e\in E\}$ and $G'=(V\cup V_E,E')$ the reduced graph. 
Since we split each edge by a vertex and $G$ is a planar cubic graph, $G'$ is not only planar but also bipartite and subcubic.
In the following, we show that there is a minimal cut of size at least $k$ in $G$ if and only if so is in $G'$.
We can assume that $k > 2$.

Let $(S_1,S_2)$ be a minimal cut of $G$.
We construct a cut $(S'_1, S'_2)$ of $G'$ with $S_i \subseteq S'_i$ for $i = 1, 2$. 
For each edge $e\in E$, we add $v_e$ to $S'_2$ if both endpoints of $e$ are contained in $S_2$, and otherwise add $v_e$ to $S'_1$.
Recall that a cut is minimal if and only if both sides of the cut induce connected subgraphs.
We claim that both $G'[S'_1]$ and $G'[S'_2]$ are connected.
To see this, consider vertices $u, v \in S_1$.
As $G[S_1]$ is connected, there is a path between $u$ and $v$ in $G[S_1]$.
By the construction of $S'_1$, every vertex of the path is in $S'_1$ and for every edge $e$ in the path, we have $v_e \in S'_1$.
Therefore, there is a path between $u$ and $v$ in $G'[S'_1]$.
Moreover, for every $v_e \in S'_1$, at least one endpoint of $e$ is in $S'_1$.
Hence, $G'[S'_1]$ is connected. Symmetrically, we can conclude that $G'[S'_2]$ is connected.
Moreover, for each $e=\{u,w\}$ with $u\in S'_1$ and $w\in S'_2$, $\{v_e,w\}$ is a cut edge in $G'$.
Therefore, $(S'_1,S'_2)$  is a minimal cut of size at least $k$.

Conversely, we are given a minimal cut $(S'_1,S'_2)$ of size $k$.
We let $S_i = S'_1 \cap V$ for $i = 1, 2$.
For each $e = \{u,v\}$, we can observe that $v_e\in S'_i$  if $u,w\in S'_i$ due to the connectivity of $S'_i$ and $k > 2$.
This means that an edge $\{u, v_e\}$ (or $\{w, v_e\}$) contributes to the cut if and only if exactly one of $u$ and $w$ is contained in $S'_1$ (and hence $S_1$), that is, the edge $e$ contributes to the cut $(S_1, S_2)$ in $G$.
Therefore, the size of the cut $(S_1, S_2)$ is at least $k$.
Moreover, $u$ and $v$ are connected by a path through $v_e$ in $G'[S'_i]$ if and only if $u$ and $v$ are contained in $S_i$ and adjacent to each other in $G[S_i]$.
Hence, $G[S_i]$ is connected for each $i = 1, 2$, and the theorem follows.
\end{proof}

\subsection{Split graphs}\label{sec:split}

\begin{figure}[tbp]
  \begin{center}
   \includegraphics[width=77mm]{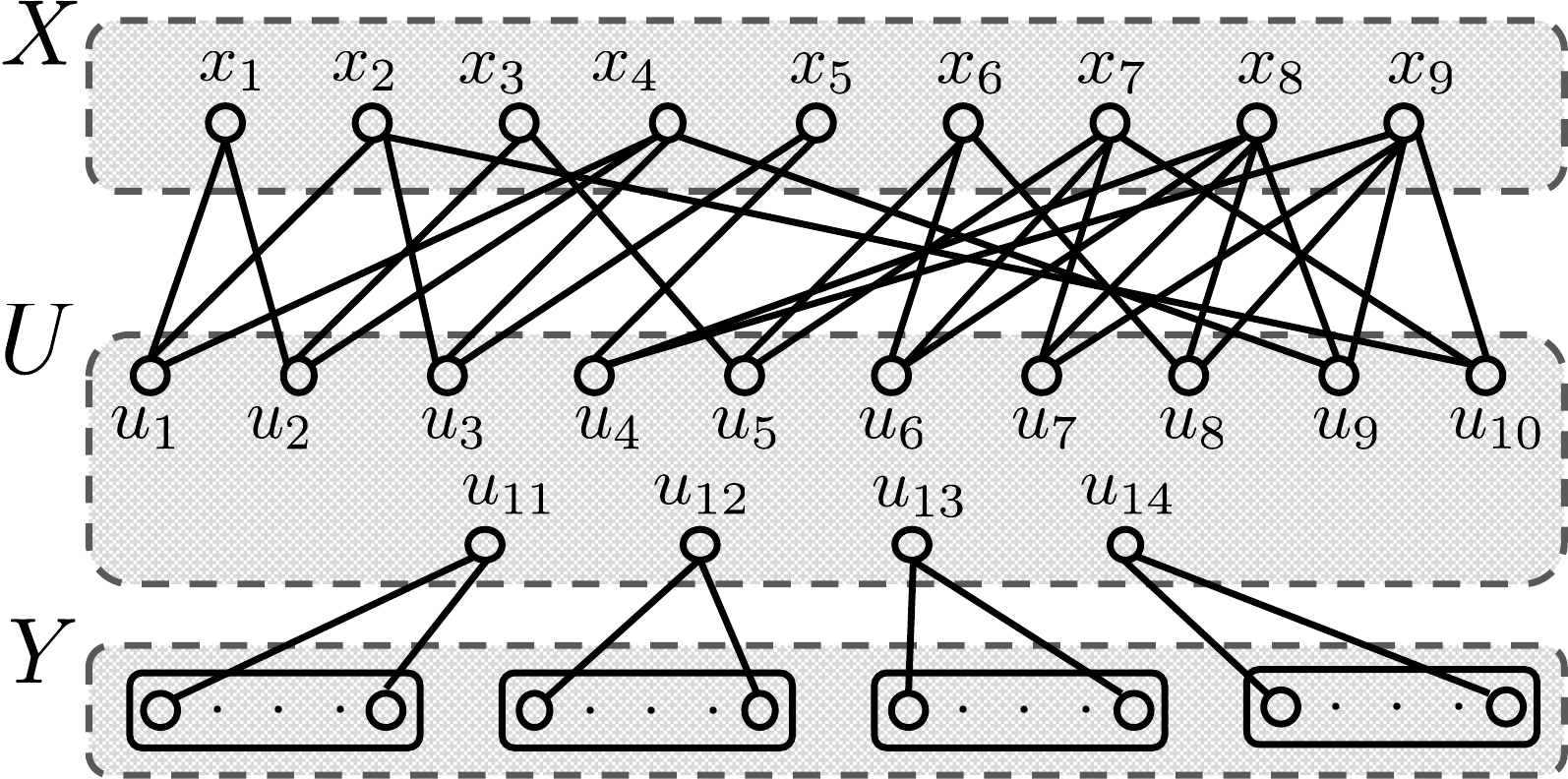}
  \end{center}
  \caption{An instance of \textsc{Connected Maximum Cut} on split graphs reduced
 from an instance of {\sc Exact 3-cover} where $X = \{x_1, x_2,x_3,x_4,x_5, x_6,x_7,x_8,
  x_9\}$ and $\mathcal F = \{
  \{x_1,x_2,x_3\},
  \{x_1,x_3,x_4\},
  \{x_2,x_4,x_5\},
  \{x_5,x_8,x_9\},
  \{x_3,x_6,x_7\},
  \{x_6,x_7,x_8\},
  \{x_7,x_8,x_9\},
  \{x_6,x_8,x_9\},\\
  \{x_4,x_8,x_9\},
  \{x_2,x_7,x_9\}
  \}$.
  }
 \label{CMConsplit}
\end{figure}

\begin{theorem}\label{thm:split:connected}
    \textsc{Connected Maximum Cut} is NP-complete on split graphs.
\end{theorem}
\begin{proof} 
We reduce the following problem called {\sc Exact 3-cover}, which is known to be NP-complete:
Given a set $X = \{x_1, x_2, \dots , x_{3n}\}$ and a family $\mathcal F = \{F_1, F_2, \dots , F_m\}$,
where each $F_i = \{x_{i_1}, x_{i_2}, x_{i_3}\}$ has three elements of $X$, 
the objective is to find a subfamily $\mathcal F' \subseteq \mathcal F$ such that 
every element in $X$ is contained in exactly one of the subsets $\mathcal F'$. 
By making some copies of $3$-element sets if necessary, 
we may assume that 
$|\{F \in \mathcal F\mid x \in F\} | \ge 3 (n+2)$ for each $x \in X$, 
which implies that $m$ is sufficiently large compared to $n$. 

Given an instance of {\sc Exact 3-cover}
with $|\{F \in \mathcal F\mid x \in F\} | \ge 3 (n+2)$ for each $x \in X$, 
we construct an instance of 
{\sc Connected Maximum Cut} in a split graph as follows. 
We introduce $m$ vertices $u_1, u_2, \dots, u_m$, where each $u_i$ corresponds to $F_i$, 
and introduce $m - 2n$ vertices $u_{m+1}, u_{m+2}, \dots, u_{2(m-n)}$. 
Let $U := \{u_1, u_2, \dots, u_{2(m-n)}\}$. 
For $i=m+1, m+2, \dots, 2(m-n)$, introduce a vertex set $Y_i$ of size $M$, 
where $M$ is a sufficiently large integer compared to $n$ (e.g. $M =3n+1$). 
Now, we construct a graph $G=(U \cup X \cup Y, E)$, where
$Y := \bigcup_{m+1 \le i \le m-2n}Y_i$,
$E_U := \{\{u, u'\} \mid u, u' \in U, u \neq u' \}$, 
$E_X := \{\{u_i, x_j\} \mid  1\le i \le m, 1 \le j \le 3n, x_j \in F_i \}$, 
$E_Y := \{\{u_i, y\} \mid m + 1 \le i \le 2(m-n), y \in Y_i \}$,  and
$E   := E_U \cup E_X \cup E_Y$. 
Then, $G$ is a split graph in which $U$ induces a clique and $X \cup Y$ is an independent set. 
We now show the following claim. 

\begin{claim*}
The original instance of {\sc Exact 3-cover} has a solution if and only if 
the obtained graph $G$
has a connected cut of size at least $(m-n)^2 + 3m - 3n + (m-2n) M$. 
\end{claim*}

\begin{proof}
Suppose that the original instance of {\sc Exact 3-cover} has a solution $\mathcal F'$. 
Then 
$S := \{u_i \mid F_i \in \mathcal F'\} \cup \{u_i \mid m+1\le i \le 2(m-n) \} \cup X$
is a desired connected cut, because 
$|\delta(S) \cap E_U| = (m-n)^2$, 
$|\delta(S) \cap E_X| = \sum_{i=1}^m |F_i| - |X| =  3m-3n$, and
$|\delta(S) \cap E_Y| = (m-2n) M$. 

Conversely, suppose that 
the obtained instance of {\sc Connected Maximum Cut} 
has a connected cut $S$ such that $|\delta(S)| \ge (m-n)^2 + 3m - 3n + (m-2n) M$. 
Since 
$|\delta(S) \cap E_U| \le (m-n)^2$, 
$|\delta(S) \cap E_X| \le  3m$, and
$|\delta(S) \cap E_Y| \le |S \cap \{ u_{m+1}, \dots , u_{2(m-n)} \}| \cdot M$,  
we obtain $|S \cap \{ u_{m+1}, \dots , u_{2(m-n)} \}| = m-2n$, that is, 
$\{ u_{m+1}, \dots , u_{2(m-n)} \} \subseteq S$. 
Let $t = |S \cap \{u_1, \dots ,u_m\}|$, $X_0 = \{x \in X \mid N(x) \cap S = \emptyset \}$ the vertices in $X$ that has no neighbor in $S$, $X_{\rm all} = \{x \in X \mid N(x) \subseteq S\}$ the vertices in $X$ whose neighbor is entirely included in $S$, and $X_{\rm part} = X \setminus (X_0 \cup  X_{\rm all})$ all the other vertices in $X$. 
Recall that every element in $X$ is contained in at least $3(n+2)$ subsets of $\mathcal F$.
Then, since 
$|\delta(S) \cap E_U| = (m - t) (m-2n+t) = (m-n)^2 - (t-n)^2$, 
$|\delta(S) \cap E_X| \le |E_X| -  |X_{\rm part}| - |\delta(X_0)|  \le  3m -  (3n - |X_{\rm all}| - |X_0|) - 3(n+2) |X_0|$, 
$|\delta(S) \cap E_Y| \le (m-2n) M$, and
$|\delta(S)| \ge (m-n)^2 + 3m - 3n + (m-2n) M$,  
we obtain 
\begin{align}\label{eq:01}
|X_{\rm all}| - (3n+5) |X_0| - (t-n)^2 \ge 0. 
\end{align}
By counting the number of edges between $S \cap \{u_1, u_2, \dots, u_m\}$ and $X$, we obtain
$3t \ge |\delta(X_{\rm all})| \ge 3(n+2) |X_{\rm all}|$, 
which shows that 
$t \ge (n+2) |X_{\rm all}|$. 
If $|X_{\rm all}| \ge 1$, then $t \ge (n+2) |X_{\rm all}| \ge n+ 2 |X_{\rm all}|$, and hence
$
|X_{\rm all}| - 3(n+5) |X_0| - (t-n)^2 
\le |X_{\rm all}| - (2 |X_{\rm all}|)^2 <0$, 
which contradicts (\ref{eq:01}). 
Thus, we obtain $|X_{\rm all}| = 0$, and hence 
we have $t=n$ and $X_0 = \emptyset$ by~(\ref{eq:01}). 
Therefore, $\mathcal F':= \{ F_i \mid  1 \le i \in m, \ u_i \in S \}$ 
satisfies that $|\mathcal F'| = n$ and $\bigcup_{F \in \mathcal F'} F = X$. 
This shows that $\mathcal F'$ is a solution of the 
original instance of {\sc Exact 3-cover}. 
\end{proof}

This shows that {\sc Exact 3-cover} is reduced to 
{\sc Connected Maximum Cut} in split graphs, 
which completes the proof. 
\end{proof}

\begin{theorem}\label{thm:split:maxmin}
\textsc{Maximum Minimal Cut} is NP-complete on split graphs.
\end{theorem}
\begin{proof}
We give a reduction from \textsc{Maximum Cut}.
Given a graph $G=(V,E)$ with $n$ vertices, we create a split graph $G'= (V\cup V_E, E')$  where  $V$ is a clique, $V_E = \{e^{\ell} \mid e \in E, 1 \le \ell \le n^3\}$ is an independent set, and $E'=\{ \{u, e^{\ell}\}, \{v,e^{\ell}\} \mid e=\{u,v\}\in E, 1 \le {\ell} \le n^3\}$. We show that $G$ has a cut of size at least $k$ if and only if $G'$ has a connected cut of size at least $kn^3$.
Without loss of generality, we assume that $n > 1$ and $k > 2$.

Let $(S_1,S_2)$ be a cut of $G$ of size $k$.
We define a cut $(S'_1, S'_2)$ of $G'$ with $S_i \subseteq S'_i$ for $i \in \{ 1, 2\}$.
For each $e\in E$ and $1 \le {\ell} \le n^3$, we set $e^{\ell} \in S'_2$ if both endpoints of $e$ are in $S_2$ in $G$, and otherwise $e^{\ell}\in S'_1$.
It is straightforward to verify that $G'[S'_1]$ and $G'[S'_2]$ are connected.
If $e = \{u, v\}$ contributes to the cut $(S_1, S_2)$, there are $n^3$ edges ($\{u, e_{\ell}\}$ or $\{v, e_{\ell}\}$) in $G'$ that contribute to $(S'_1, S'_2)$. Therefore, the size of $(S'_1, S'_2)$ is at least $kn^3$.

Conversely, we are given a minimal cut $(S'_1,S'_2)$ of size  $kn^3$ in $G'$. Let $S_i = S'_i \cap V$ for $i \in \{ 1, 2\}$. As with the proof of Theorem~\ref{thm:bipartite:maxmin}, for each $e=\{u,w\}$ and $i\in \{1,2\}$, we can observe that $e^{\ell}\in S_i$ if $u,w\in S_i$ due to the connectivity of $S_i$ and $k > 2$.
Since $V$ forms a clique in $G'$, there are at most $n^2$ cut edges in $G[V]$.
Thus, at least $kn^3-n^2 > kn^3 - n^3 = (k - 1)n^3$ edges between $V$ and $V_E$ belong to the cutset.
This implies that there are at least $k$ pairs $\{u, w\}$ with $u \in S'_1 \cap V$ and $w \in S'_2 \cap V$, and hence $G$ has a cut of size at least $k$.
\end{proof}

\section{Parameterized Complexity}

\subsection{Tree-width}\label{sec:treewidth}

In this section, we give FPT algorithms for \textsc{Connected Maximum Cut}  and \textsc{Maximum Minimal Cut} parameterized by tree-width.
In particular, we design $O^*(c^{\tw})$-time algorithms where $c$ is some constant.

\subsubsection{$O^*({\tw}^{O(\tw)})$-algorithm}
\label{sec:tw^tw}
We design an $O^*({\tw}^{O(\tw)})$-algorithm for \textsc{Maximum Minimal Cut}.
To do this, we consider a slightly different problem, called \textsc{Maximum Minimal $s$-$t$ Cut}: Given a graph $G=(V,E)$, an integer $k$ and two vertices $s,t\in V$, determine whether there is a cut $(S_1,S_2)$ of size at least $k$ in $G$ such that $s\in S_1$, $t\in S_2$ and $(S_1,S_2)$ is minimal, that is, both $G[S_1]$ and  $G[S_2]$ are connected.
If we can solve {\sc Maximum Minimal $s$-$t$ Cut} in time $O^*({\tw}^{O(\tw)})$, we can also solve \textsc{Maximum Minimal Cut} in the same running time up to a polynomial factor in $n$ since it suffices to compute \textsc{Maximum Minimal $s$-$t$ Cut} for each pair of $s$ and $t$.

Our algorithm is based on standard dynamic programming on a tree decomposition. This algorithm outputs a maximum minimal cut $(S_1, S_2)$.
Basically, the algorithm is almost the same as an $O^*(2^{\tw})$-algorithm for \textsc{Max Cut} in \cite{bodlaender2000} except for keeping the connectivity of a cut.
To keep track of the connectivity, for each bag $X_i$, we consider two partitions ${\mathcal S}_1$ and ${\mathcal S}_2$ of $S_1\cap X_i$ and $S_2\cap X_i$, respectively.
We call each set in a partition a {\em block}.

Let $i$ be a node in $T$.
Then we define a  {\em partial solution} of \textsc{Maximum Minimal Cut} at node $i$.

\begin{definition}\label{def:tw:partial_solution}
Let $S_i, T_i \subseteq X_i$ with $S_i \cap T_i = \emptyset$ and $S_i \cup T_i = X_i$ and
let $\mathcal S_i$ (resp. $\mathcal T_i$) be a partition of $S_i$ (resp. a partition of $T_i$).
A \emph{partial solution} for $(S_i, S_2, \mathcal S_i, \mathcal T_i)$ is a cut $(S, V_i \setminus S)$ of $G_i$ such that
\begin{itemize}
	\item $S \cap X_i = S_i$ and $(V_i \setminus S) \cap X_i = T_i$, 
	\item $\forall u, v\in S_i$ (resp. $\forall u, v \in T_i$), $u$ and $v$ are in the same block in $\mathcal S_i$ (resp. in $\mathcal T_i$) $\iff$ there is a path between $u$ and $v$ in $G_i[S_i]$ (resp. in $G_i[T_i]$), and
	\item $s\in V_i \implies s\in S_i$ and $t\in V_i \implies t\in T_i$.
\end{itemize}
\end{definition}

For each $S_i, T_i \subseteq X_i$ and for $\mathcal S_i$ and $\mathcal T_i$ as in Definition~\ref{def:tw:partial_solution},
we compute the value $\mc_i(S_i, T_i, \mathcal S_i, \mathcal T_i)$ that is the maximum size of a partial solution for $(S_i, T_i, \mathcal S_i, \mathcal T_i)$ of $G_i$.

We now define recursive formulas for computing each row $\mc_i(S_i, T_i, {\mathcal S}_i, {\mathcal T}_i)$ at node $i$ on a nice tree decomposition. 
First, we add $s,t\in V$ to each bag of $T$ and remove the bags introduce $s$ and $t$ by connecting its parent and child directly.
Thus, the root and leaf bags satisfy $X_{r(T)}=\{s,t\}$ and the width of this tree decomposition increases by at most two. 
Notice that our goal is to compute $\mc_r(\{s\}, \{t\}, \{\{s\}\}, \{\{t\}\})$.
We denote the current node by $i$ and its child node by $j$ except for leaf and join nodes. For a join node, we write $j_1$ and $j_2$ to denote its two children. 

\paragraph*{Leaf node:}  
In a leaf node $i$, the bag contain only $s$ and $t$. As $G_i$ contains only two isolated vertices $s$ and $t$, $\mc_i(\{s\}, \{t\}, \{\{s\}\}, \{\{t\}\}) = 0$. For any other $S_i, T_i, \mathcal S_i, \mathcal T_i$, $\mc_i(S_i, T_i, \mathcal S_i, \mathcal T_i) = -\infty$.

\paragraph*{Introduce vertex $v$ node:} 
In an introduce vertex $v$ node $i$, we have two choices: $v$ is included in either $S_i$ or $T_i$. 
Note that $v$ must be an isolated vertex in $G_i$ since edges incident to $v$ have not yet been introduced.
Then, the recursive formula is defined as: 
\[
\mc_i(S_i, T_i, {\mathcal S}_1, {\mathcal S}_2):= 
\begin{cases}
    \mc_{j}(S_i \setminus \{v\}, T_i, {\mathcal S}_i\setminus \{\{v\}\}, {\mathcal T}_i) &  \mbox{if $v \in S_i$ and $\{\{v\}\}\in {\mathcal S}_i$},\\
    \mc_{j}(S_i, T_i \setminus \{v\}, {\mathcal S}_i, {\mathcal T}_i\setminus \{\{v\}\}) &  \mbox{if $v \in T_i$ and $\{\{v\}\}\in {\mathcal T}_i$},\\
 -\infty & \mbox{otherwise}.
\end{cases}
\]

\paragraph*{Introduce edge $\{u, v\}$ node:} 
In an introduce edge node $i$, we have the following cases.

\begin{description}
\item[Case: either $u, v \in S_i$ or $u, v \in T_i$]
\

Suppose that $u, v \in S_i$.
If $u$ and $v$ are in the different blocks in ${\mathcal S}_i$, $\mc_i[S_i, T_i, {\mathcal S}_i, {\mathcal T}_i] := -\infty$ since $u$ and $v$ should be in the same block of partition ${\mathcal S}_i$ due to the existence of the edge $\{u, v\}$ in $G_i$.  

Otherwise, let $\Pi(\mathcal S_i, u, v)$ be the set of partitions of $S_i$ such that each partition $\mathcal S$ contains distinct two blocks containing $u$ and $v$ and $\mathcal S_i$ is obtained from $\mathcal S$ by merging the two blocks containing $u$ and $v$.
Then, 
\[
    \mc_i(S_i, T_i, {\mathcal S}_i, {\mathcal T}_i) := \max \{\mc_j(S_i, T_i, {\mathcal S}_i, {\mathcal T}_i), \max_{{\mathcal S} \in \Pi(\mathcal S_i, u, v)} \mc_j(S_i, T_i, {\mathcal S}, {\mathcal T}_i)\}.
\]

Symmetrically, if $u, v \in T_i$, we have
\[
    \mc_i(S_i, T_i, {\mathcal S}_i, {\mathcal T}_i) := \max \{\mc_j(S_i, T_i, {\mathcal S}_i, {\mathcal T}_i), \max_{{\mathcal T} \in \Pi(\mathcal T_i, u, v)} \mc_j(S_i, T_i, {\mathcal S}_i, {\mathcal T})\}.
\]

\item[Case: $|S_i \cap \{u, v\}| = |T_i \cap \{u, v\}| = 1$]
\

In this case, $\{u,v\}$ is included in the cutset and it does not involve the partitions ${\mathcal S}_i, {\mathcal T}_i$.
Thus, the recursive formula is defined as: 
\[
    \mc_i(S_i, T_i, {\mathcal S}_i, {\mathcal T}_i) := \mc_j(S_i, T_i, {\mathcal S}_i, {\mathcal T}_i) + 1.
\]
\end{description}

\paragraph*{Forget $v$ node:}
Let $S_j, T_j \subseteq X_j$ with $S_j \cap T_j = \emptyset$ and $S_j \cup T_j = X_j$.
Let $\mathcal S_j$ and $\mathcal T_j$ be arbitrary partitions of $S_j$ and $T_j$, respectively.
Suppose that vertex $v\in S_j$ is in a singleton block of $\mathcal S_j$, that is, $\{\{v\}\}\in {\mathcal S}_j$.
Every edge incident to $v$ has already been introduced at some descendant node of $j$.
This means that there is no path between $v$ and $s$ in $G[S]$ of any partial solution $(S, V_k \setminus S)$ at any ancestor node $k$ of $i$.
Therefore,  $v$ is contained in a block of size at least two in node $j$.
Thus, the recursive formula is defined as: 
\begin{align*}
    \mc_i(S_i, T_i, {\mathcal S}_i, {\mathcal T}_i):=
        \max \left\{\max_{\mathcal S}\mc_{j}(S_i \cup \{v\}, T_i, \mathcal S, {\mathcal T}_i),
        \max_{\mathcal T}\mc_{j}(S_i, T_i \cup \{v\}, \mathcal S_i, \mathcal T) \right\}
\end{align*}
where, in the first case, the maximum is taken among all partitions $\mathcal S$ of $S_j=S_i\cup \{v\}$, each of which is obtained from $\mathcal S_i$ by adding $v$ to one of the existing blocks, and
in the second case, the maximum is taken among all partitions $\mathcal T$ of $T_j=T_i\cup \{v\}$, each of which is obtained from $\mathcal T_i$ by adding $v$ to  one of the existing blocks.

\paragraph*{Join node:}
Let $\mathcal X_1$ and $\mathcal X_2$ be two partitions of the same set $X$.
We denote by $J(\mathcal X_1, \mathcal X_2)$ be a partition of $X$ that is obtained by joining two partitions $\mathcal X_1$ and $\mathcal X_2$. In other words, $J(\mathcal X_1, \mathcal X_2)$ is the finest partition of $X$ that is coarser than both $\mathcal X_1$ and $\mathcal X_2$.

For join node $i$ with children $j_1$ and $j_2$, we have
\begin{align*}
\mc_i[S_i, T_i, {\mathcal S}_i, {\mathcal T}_i] :=  \max_{\substack{\mathcal S_{j_1}, \mathcal S_{j_2}, \mathcal T_{j_1}, \mathcal T_{j_2}:\\
\mathcal S_i = J(\mathcal S_{j_1}, \mathcal S_{j_2}),\ \mathcal T_i = J(\mathcal T_{j_1}, \mathcal T_{j_2})}} 
\{
&\mc_{j_1}(S_{j_1}, T_{j_1}, {\mathcal S}_{j_1}, {\mathcal T}_{j_1}) + \mc_{j_2}(S_{j_2}, T_{j_2}, {\mathcal S}_{j_2}, {\mathcal T}_{j_2})
 \}, 
\end{align*}
where $S_i = S_{j_1} = S_{j_2}$ and $T_i = T_{j_1} = T_{j_2}$.

Clearly, in each node, we can compute the recursive formulas in time ${\tw}^{O(\tw)}$.
Therefore, the total running time is $O^*({\tw}^{O(\tw)})$.
The correctness of the formulas are similar to ones for other connectivity problems, and hence we omit the proof here. 
As we said, for \textsc{Maximum Minimal Cut}, it suffices to apply the algorithm to each pair of $s$ and $t$.

\begin{theorem}\label{thm:mmc:tw^tw}
Given a tree decomposition of width $\tw$, \textsc{Maximum Minimal $s$-$t$ Cut} and \textsc{Maximum Minimal Cut} are solvable in time $O^*({\tw}^{O(\tw)})$, respectively.
\end{theorem}

We can obtain a similar algorithm for \textsc{Connected Maximum Cut}: we do not have to take care of the connectivity information for $S_2$ and simply drop it in the above computation.
Thus, we have the following theorem.
\begin{theorem}\label{thm:cmc:tw^tw}
Given a tree decomposition of width $\tw$, \textsc{Connected Maximum $s$-$t$ Cut} and \textsc{Connected Maximum Cut} are solvable in time $O^*({\tw}^{O(\tw)})$, respectively.
\end{theorem}




The dynamic programming algorithms in Theorems\ref{thm:mmc:tw^tw} and \ref{thm:cmc:tw^tw} can be seen as ones for {\em connectivity problems} such as finding a Hamiltonian cycle, a feedback vertex set, and a Steiner tree.
For such problems, we can improve the running time $\tw^{O(\tw)}$ to $2^{O(\tw)}$ using two techniques called the {\em rank-based approach} due to Bodlaender et al.~\cite{Bodlaender2015} and the {\em cut \& count technique} due to Cygan et al.~\cite{Cygan2011}. In the next two subsections, we improve the running time of the algorithms described in this section using these techniques.

\subsubsection{Rank-based approach}

In this subsection, we provide  faster $2^{O(\tw)}$-time deterministic algorithms parameterized by tree-width. To show this, we use the rank-based approach proposed by Bodlaender et al. \cite{Bodlaender2015}. 
The key idea of the rank-based approach is to keep track of {\em small} representative sets of size $2^{O(\tw)}$ that capture partial solutions of an optimal solution instead of $\tw^{O(\tw)}$ partitions.
Indeed, we can compute small representative sets within the claimed running time using {\sf reduce} algorithm \cite{Bodlaender2015}.

We begin with some definition used in the Rank-based approach.
\setcounter{definition}{29}
\begin{definition}[Set of weighted partitions \cite{Bodlaender2015}]
Let $\Pi(U)$ be the set of all partitions of some set $U$. A set of weighted partitions is ${\mathcal A} \subseteq \Pi(U)\times {\mathbb N}$, i.e., a family of pairs, each consisting of a partition of $U$ and a non-negative integer weight. 
\end{definition}

The weight of a partition corresponds to the size of a partial solution.
For $p,q\in \Pi(U)$, let $J(p, q)$ denote the join of the partition. We say that a set of weighted partitions ${\mathcal A}' \subseteq \Pi(U)\times {\mathbb N}$ {\em represents} another set ${\mathcal A} \subseteq \Pi(U)\times {\mathbb N}$, if for all $q\in \Pi(U)$ it holds that  $\max\{w\mid (p,w)\in {\mathcal A}' \land J(p, q) = \{U\}\}=\max\{w\mid (p,w)\in {\mathcal A} \land J(p, q)=\{U\}\}$.
Then Bodlaender et al. \cite{Bodlaender2015} provided {\sf reduce} algorithm that computes a {\em small} representative set of weighted partitions.
\begin{theorem}[{\sf reduce} \cite{Bodlaender2015}]\label{thm:tw:reduce}
There exists an algorithm {\sf reduce} that given a set of weighted partitions ${\mathcal A} \subseteq \Pi(U)\times {\mathbb N}$, outputs in time $|{\mathcal A}|2^{(\omega-1)|U|}|U|^{O(1)}$ a set of weighted partitions ${\mathcal A}'\subseteq {\mathcal A}$ such that ${\mathcal A}'$ represents ${\mathcal A}$ and $|{\mathcal A}| \le 2^{|U|-1}$, where $\omega < 2.3727$ denotes the matrix multiplication exponent.
\end{theorem}
The {\sf reduce} algorithm allows us to compute an optimal solution without keeping all weighted partitions.
We apply {\sf reduce} algorithm to the set of partitions at each node in the $O^*({\tw}^{O(\tw)})$-time algorithm in the previous section.

\begin{theorem}\label{thm:rank-based:concut}
Given a tree decomposition of width $\tw$, there are $O^*((1+2^{\omega+1})^{\tw})$-time deterministic algorithms for \textsc{Connected Maximum $s$-$t$ Cut} and \textsc{Connected Maximum Cut}.
\end{theorem}
\begin{proof}
For a bag $X_i$, we compute the value $\mc_i(S_i, T_i, \mathcal S_i)$ for each $S_i, T_i \subseteq X_i$ with $S_i \cap T_i = \emptyset$ and $S_i \cup T_i = X_i$ and $\mathcal S_i$ is a partition of $S_i$.
For each $S_i$ and $T_i$,
we apply the {\sf reduce} algorithm to a set of weighted partitions $(\mathcal S_i, \mc_i(S_i, T_i, \mathcal S_i))$ that are obtained by the recursive formulas described in the previous section.
At each node $i$, the {\sf reduce} algorithm outputs only $2^{|S_i| - 1}$ weighted partitions for each $S_i$.
Thus, at each node except join nodes, the running time of evaluating the recursive formula is $\sum_{S_i \subseteq X_i}O^*(2^{|S_i|}) = O^*(3^{\tw})$ and of the {\sf reduce} algorithm is $\sum_{S_i \subseteq X_i}O^*(2^{|S_i|}2^{(\omega - 1)|S_i|}) = \sum_{S_i \subseteq X_i}O^*(2^{\omega|S_i|}) = O^*((1 + 2^{\omega})^\tw)$.
At each join node, since the output of evaluating the recursive formula may contain $O^*(2^{2|S_i|})$ weighted partitions for each $S_i$. 
Thus, the total running time at join node $i$ is $\sum_{S_i \subseteq X_i}O^*(2^{2|S_i|}2^{(\omega - 1)|S_i|}) = O^*((1 + 2^{\omega + 1})^\tw)$.
Hence, the theorem follows.
\end{proof}

Note that if a tree decomposition has no join nodes, namely a path decomposition, the overall running time is $O^*((1 + 2^{\omega})^\pw)$.

\begin{theorem}\label{thm:rank-based:minimalcut}
Given a tree decomposition of width $\tw$, there are $O^*(2^{(\omega +2)\tw})$-time deterministic algorithms for \textsc{Maximum Minimal  $s$-$t$  Cut} and \textsc{Maximum Minimal Cut}.
\end{theorem}
\begin{proof}
For a bag $X_i$, we compute the value $\mc_i(S_i, T_i, \mathcal S_i, \mathcal T_i)$ for each $S_i, T_i \subseteq X_i$ with $S_i \cap T_i = \emptyset$ and $S_i \cup T_i = X_i$ and $\mathcal S_i, \mathcal T_i$ are partitions of $S_i$ and $T_i$, respectively.
Similar to Theorem~\ref{thm:rank-based:concut}, for each $S_i, T_i$, and $\mathcal T_i$,
we apply the {\sf reduce} algorithm to a set of weighted partitions $(\mathcal S_i, \mc_i(S_i, T_i, \mathcal S_i, \mathcal T_i))$ and then apply it again to weighted partitions $(\mathcal T_i, \mc_i(S_i, T_i, \mathcal S_i, \mathcal T_i))$ for each $S_i, T_i$, and for each remaining $\mathcal S_i$ of the first application.
Since there are at most $2^{|S_i|-1}2^{|T_i|-1} = 2^{|X_i|-2}$ weighted partitions in the representative set for each $S_i \subseteq X_i$,
the total running time is $\sum_{S_i \subseteq X_i}O^*(2^{2(|X_i|-2)}2^{(\omega-1)|X_i|}) = O^*(2^{(\omega + 2)\tw})$.
\end{proof}

\subsubsection{{Cut \& Count}}

In this subsection, we design much faster randomized algorithms by using Cut \& Count, which is the framework for solving the connectivity problems faster~\cite{Cygan2011}. 
In Cut \& Count,  we count the number of {\em relaxed} solutions modulo 2 on a tree decomposition and determine whether there exists a connected solution by cancellation tricks.

\begin{definition}[\cite{Cygan2011}]
A cut $(V_1,  V_2)$ of $V' \subseteq V$ such that $V_1\cup V_2 = V'$ and $V_1\cap V_2 = \emptyset$ is {\em consistent} if $v_1 \in V_1$ and $v_2 \in V_2$ implies $(v_1, v_2) \notin E$. 
\end{definition}
In other words, a cut $(V_1, V_2)$ of $V'$ is consistent if there are no cut edges between $V_1$ and $V_2$. 
 
Fix an arbitrary vertex $v$ in $V_1$. If $G[V]$ has $k$ components, then there exist $2^{k-1}$ consistent cuts of $V$. Thus, when $G[V]$ is connected, there only exists one consistent cut $(V_1, V_2)=(V, \emptyset)$. From this observation, $G[V]$ is connected if and only if the number of consistent cuts is odd. 
Therefore, in order to compute ``connected solutions'', it seems to suffice to count the number of consistent cuts modulo two at first glance.
However, this computation may fail to count the number of ``connected solutions'' since there can be even number of such solutions. 
To overcome this obstacle, Cygan et al.~\cite{Cygan2011} used the Isolation Lemma~\cite{MVV1987}, which ensures with high probability that the problem has a unique minimum solution.
For the detail of the Isolation Lemma, see~\cite{Cygan2015,MVV1987}.



We follow the Cut \& Count framework in \cite{Cygan2015,Cygan2011}: We apply it to determining whether there exists a minimal $s$-$t$ cut, a cut that separates $s$ and $t$, of size $k$, namely \textsc{Maximum Minimal $s$-$t$ Cut}.
Recall that $(S, V \setminus S)$ is a minimal $s$-$t$ cut of a connected graph $G = (V, E)$ if both $G[S]$ and $G[V\setminus S]$ are connected, $s \in S$, and $t \in V \setminus S$.

Let $i$ be a node of a nice tree decomposition of $T$ of $G$.

\begin{definition}
Let $\ell \in \{1, 2, \ldots, |E|\}$. Let $S^l_i, S^r_i, T^l_i, T^r_i$ be pairwise disjoint (possibly) subsets of $X_i$ such that $S^l_i \cup S^r_i \cup T^l_i \cup T^r_i = X_i$.
A \emph{partial solution} for $(S^l_i, S^r_i, T^l_i, T^r_i, \ell)$ is a cut $(S, V_i \setminus S)$ of $G_i$ such that:
\begin{itemize}
	\item $S \cap X_i = S^l_i \cup S^r_i$ and $(V_i \setminus S) \cap X_i = T^l_i \cup T^r_i$,
	\item $(S_i^l,S_i^r)$ and $(T^l_i, T^r_i)$ are consistent cuts of $S^l_i \cup S^r_i$ and $T^l_i \cup T^r_i$, respectively, 
	\item there are exactly $\ell$ cut edges between $S$ and $V_i \setminus S$ in $G_i$, and 
	\item $s\in V_i \implies s\in S_i^l$ and $t\in V_i \implies t\in T_i^l$.
\end{itemize}
\end{definition}

Before proceeding to our dynamic programming, we assign a weight $w_v$ to each vertex $v \in V$ by choosing an integer from $\{1, \ldots, 2n\}$  independently and uniformly at random.
We also use the preprocessing used in \revised{an $O^*({\tw}^{O(\tw)})$-algorithm}: add $s$ and $t$ to each node of $T$ and remove the bags introduce $s$ or $t$ from $T$.
In our dynamic programming algorithm, for each node $i$ and for $0 \le w \le 2n^2$ and $0\le \ell \le |E|$, 
we count the number of partial solutions $(S, V_i\setminus S)$ for $(S^l_i, S^r_i, T^l_i, T^r_i, \ell)$ such that the total weight of $S$ is exactly $w$, which we denote by $c(S^l_i, S^r_i, T^l_i, T^r_i, \ell, w)$.
By the Isolation Lemma, with high probability, there is a minimal $s$-$t$ cut of $G$ of size exactly $k$ if and only if $c(\{s\}, \emptyset, \{t\}, \emptyset, k, w)$ is odd 
for some $0 \le w \le 2n^2$ in the root node $r(T)$.
In the following, we describe the recursive formula for our dynamic programming.

\paragraph*{Leaf node:}  In a leaf node $i$, since $X_i= \{s, t\}$, we have $c_i(S^l_i, S^r_i, T^l_i, T^r_i, \ell, w)=1$  if $S_i^l=\{s\}$, $T^l_i=\{t\}$, $S_{1}^r=T^r_i=\emptyset$, $\ell=0$, and $w=0$. 
Otherwise, $c_i(S^l_i, S^r_i, T^l_i, T^r_i, \ell, w)=0$.

\paragraph*{Introduce vertex  $v$ node:} 
In an introduce vertex node $i$, we consider the following four cases:
\begin{eqnarray*}
c_i(S^l_i, S^r_i, T^l_i, T^r_i, \ell, w)
& := \begin{cases}
c_j(S^l_i\setminus \{v\}, S^r_i, T^l_i, T^r_i, \ell, w-w(v))& \mbox{if  $v\in S^l_i$,} \\
c_j(S^l_i, S^r_i\setminus \{v\},, T^l_i, T^r_i, \ell, w-w(v))& \mbox{if  $v\in S^r_i$,} \\
c_j(S^l_i, S^r_i, T^l_i\setminus \{v\},, T^r_i, \ell, w)& \mbox{if  $v\in T^l_i$,} \\
c_j(S^l_i, S^r_i, T^l_i, T^r_i\setminus \{v\}, \ell, w)& \mbox{if $v \in T^r_i$}. \\
  \end{cases}
\end{eqnarray*}
As $v \in X_i$, exactly one of the above cases is applied. 

\paragraph*{Introduce edge $(u, v)$ node:} 
Let $i$ be an introduce node of $T$.
Let $S^l_i, S^r_i, T^l_i, T^r_i$ be disjoint subsets of $X_i$ whose union covers $X_i$.
If exactly one of $u$ and $v$ belongs to $S^l_i\cup S^r_i$ (i.e. the other one belongs to $T^l_i\cup T^r_i$),
the edge is included in the cutset.
Suppose otherwise, that is, either $u, v \in S^l_i \cup S^r_i$ or $u, v \in T^l_i \cup T^r_i$. 
If $u$ and $v$ belong to different sets, say $u \in S^l_i$ and $v \in S^r_i$, then
$(S^l_i, S^r_i)$ is not consistent.
Therefore, there is no partial solutions in this case.
To summarize these facts, we have the following:
\begin{eqnarray*}
c_i(S^l_i, S^r_i, T^l_i, T^r_i, \ell, w)
& := \begin{cases}
	c_j(S^l_i, S^r_i, T^l_i, T^r_i, \ell-1, w)& \mbox{if $|(S^l_i\cup S^r_i) \cap \{u, v\}| = 1$,} \\
	c_j(S^l_i, S^r_i, T^l_i, T^r_i, \ell, w) & \mbox{if  $u,v$ are in the same set,}\\
	0 & \mbox{otherwise}.
  \end{cases}
\end{eqnarray*}

\paragraph*{Forget $v$ node:}
In a forget node $i$, we just sum up the number of partial solutions:
\begin{eqnarray*} 
c_i(S^l_i, S^r_i, T^l_i, T^r_i, \ell, w) := &c_j(S^l_i\cup \{v\}, S^r_i, T^l_i, T^r_i, \ell, w) + c_j(S^l_i, S^r_i\cup \{v\}, T^l_i, T^r_i, \ell, w)\\
&+c_j(S^l_i, S^r_i, T^l_i\cup \{v\}, T^r_i, \ell, w) +c_j(S^l_i, S^r_i, T^l_i, T^r_i\cup \{v\}, \ell, w).
\end{eqnarray*}

\paragraph*{Join node:}
Let $i$ be a join node and $j_1$ and $j_2$ its children.
As $X_i = X_{j_1} = X_{j_2}$, it should hold that $S^l_i = S^l_{j_1} = S^l_{j_2}$, $S^r_i = S^r_{j_1} = S^r_{j_2}$, $T^l_i = T^l_{j_1} = T^l_{j_2}$, and $T^r_i = T^r_{j_1} = T^r_{j_2}$.
The size of a partial solution $S_i$ at $i$ is the sum of the size of partial solutions $S_{j_1}$ and $S_{j_2}$ at its children and also the total weight of $S$ is the the sum of the weight of $S_{j_1}$ and $S_{j_2}$.
Thus, we have
\begin{eqnarray*} 
c_i(S^l_i, S^r_i, T^l_i, T^r_i, \ell, w) := \sum_{\ell_{j_1}+\ell_{j_2}=\ell}   \sum_{w_{j_1}+w_{j_2}=w} c_{j_1}(S^l_i, S^r_i, T^l_i, T^r_i, \ell_{j_1}, w_{j_1}) c_{j_2}(S^l_i, S^r_i, T^l_i, T^r_i, \ell_{j_2}, w_{j_2}) .
\end{eqnarray*}

The running time of evaluating the recursive formulas is $O^*(4^{|X_i|})$ for each node $i$.
Therefore, the total running is $O^*(4^{\tw})$.
We can also solve  \textsc{Maximum Minimal Cut} in time $O^*(4^{\tw})$ by applying the algorithm for \textsc{Maximum Minimal $s$-$t$ Cut} for all combinations of $s$ and $t$.

\begin{theorem}\label{thm:treewidth_single:Maxmin}
Given a tree decomposition of width $\tw$, there is a Monte-Carlo algorithm that solves \textsc{Maximum Minimal Cut} and \textsc{Maximum Minimal $s$-$t$ Cut}  in time $O^*(4^{\tw})$. It cannot give false positives
and may give false negatives with probability at most 1/2.
\end{theorem}

We can also solve  \textsc{Connected Maximum Cut} and \textsc{Connected Maximum $s$-$t$ Cut}.
Since it suffices to keep track of consistent cuts of $S$, the running time is $O^*(3^{\tw})$.
\begin{theorem}\label{thm:treewidth_single:Connected}
Given a tree decomposition of width $\tw$, there is a Monte-Carlo algorithm that solves  \textsc{Connected Maximum Cut} and \textsc{Connected Maximum $s$-$t$ Cut}  in time $O^*(3^{\tw})$. It cannot give false positives
and may give false negatives with probability at most 1/2.
\end{theorem}

\subsection{Clique-width}
In this section, we design XP algorithms for both \textsc{Connected Maximum Cut} and \textsc{Maximum Minimal Cut} when parameterized by clique-width.
The algorithms are analogous to the dynamic programming algorithm for \textsc{Maximum Cut} given by Fomin et al.~\cite{Fomin2014}, but we need to carefully control the connectivity information in partial solutions.

Suppose that the clique-width of $G$ is $w$. Then, $G$ can be constructed by the four operations described in Definition~\ref{def:clique-width}.
This construction naturally defines a tree expressing a sequence of operations. This tree is called a {\em $w$-expression tree} of $G$ and used for describing dynamic programming algorithms for many problems based on clique-width.
Here, we rather use a different graph parameter and its associated decomposition closely related to clique-width.
We believe that this decomposition is more suitable to describe our dynamic programming.

\begin{definition}
    Let $X \subseteq V(G)$. We say that $M \subseteq X$ is a {\em twin-set} of $X$ if for any $v \in V(G) \setminus X$, either $M \subseteq N(v)$ or $M \cap N(v) = \emptyset$ holds. A twin-set $M$ is called a {\em twin-class} of $X$ if it is maximal subject to being a twin-set of $X$. $X$ can be partitioned into twin-classes of $X$.
\end{definition}

\begin{definition}
    Let $w$ be an integer. We say that $X \subseteq V(G)$ is a {\em $w$-module} of $G$ if $X$ can be partitioned into $w$ twin-classes $\{X_1, X_2, \ldots, X_{w}\}$.
    A {\em decomposition tree} of $G$ is a pair of a rooted binary tree $T$ and a bijection $\phi$ from the set of leaves of $T$ to $V(G)$.
    For each node $v$ of $T$, we denote by $L_v$ the set of leaves, each of which is either $v$ or a descendant of $v$.
    The {\em width} of a decomposition tree $(T, \phi)$ of $G$ is the minimum $w$ such that for every node $v$ in $T$, the set $\bigcup_{l \in L_v}\phi(l)$ is a $w_v$-module of $G$ with $w_v \le w$.
    The {\em module-width} of $G$ is the minimum $t$ such that there is a decomposition tree of $G$ of width $w$.
\end{definition}

Rao \cite{Rao2008} proved that clique-width and module-width are linearly related to each other.
Let $\cw(G)$ and $\mw(G)$ be the clique-width and the module-width of $G$, respectively.
We note that a similar terminology ``modular-width'' has been used in many researches, but module-width used in this paper is different from it.

\begin{theorem}[\cite{Rao2008}]\label{thm:mw-cw}
    For every graph $G$,
    $\mw(G) \le \cw(G) \le 2\mw(G)$.
\end{theorem}
Moreover, given a $w$-expression tree of $G$, we can in time $O(n^2)$ compute a decomposition tree $(T, \phi)$ of $G$ of width at most $w$ and $w_v \le w$ twin-classes of $\bigcup_{l \in L_v}\phi(l)$ for each node $v$ in $T$ \cite{Bui-Xuan2013}.

Fix a decomposition tree $(T, f)$ of $G$ whose width is $w$.
Our dynamic programming algorithm runs over the nodes of the decomposition tree in a bottom-up manner.
For each node $v$ in $T$, we let $\{X^v_1, X^v_2, \ldots, X^v_{w_v}\}$ be the twin-classes of $\bigcup_{l \in L_v}\phi(l)$. From now on, we abuse the notation to denote $\bigcup_{l \in L_v}\phi(l)$ simply by $L_v$.
A tuple of $4w_v$ integers $t= (p_1, \overline{p}_1, p_2, \overline{p}_2, \ldots, p_{w_v}, \overline{p}_{w_v}, c_1, \overline{c}_1, c_2, \overline{c}_2, \ldots, c_{w_v}, \overline{c}_{w_v})$ is {\em valid} for $v$ if it holds that $0 \le p_i, \overline{p}_i \le |X^v_i|$ with $p_i + \overline{p}_i = |X^v_i|$ and $c_i, \overline{c}_i \in \{0, 1\}$ for each $1 \le i \le w_v$.
For a valid tuple $t$ for $v$, we say that a cut $(S, L_v \setminus S)$ of $G[L_v]$ is {\em $t$-legitimate} if for each $1 \le i \le w_v$, it satisfies the following conditions:
\begin{itemize}
    \item $p_i = |S \cap X^v_i|$,
    \item $\overline{p}_i = |(L_v \setminus S) \cap X^v_i|$,
    \item $G[S \cap X^v_i]$ is connected if $c_i = 1$, and
    \item $G[(L_v \setminus S) \cap X^v_i]$ is connected if $\overline{c}_i = 1$.
\end{itemize}
The size of a $t$-legitimate cut is defined accordingly.
In this section, we allow each side of a cut to be empty and the empty graph is considered to be connected.
Our algorithm computes the value $\mc(v, t)$ that is the maximum size of a $t$-legitimate cut for each valid tuple $t$ and for each node $v$ in the decomposition tree. 

\paragraph*{Leaves (Base step):} 
For each valid tuple $t$ for a leaf $v$, $\mc(v, t) = 0$. 
Note that there is only one twin-class $X^v_1 = \{v\}$ for $v$ in this case. 

\paragraph*{Internal nodes (Induction step):} Let $v$ be an internal node of $T$ and let $a$ and $b$ be the children of $v$ in $T$.
Consider twin-classes $\CX^v = \{X^v_1, X^v_2, \ldots, X^v_{w_v}\}$, $\CX^a = \{X^a_1, X^a_2, \ldots, X^a_{w_a}\}$, and $\CX^b = \{X^b_1, X^b_2, \ldots, X^b_{w_b}\}$ of $L_v$, $L_a$, and $L_b$, respectively. Note that $\CX^a \cup \CX^b$ is a partition of $L_v$.

\begin{observation}\label{obs:coarse}
    $\CX^v$ is a partition of $L_v$ coarser than $\CX^a \cup \CX^b$.
\end{observation}

To see this, consider an arbitrary twin-class $X^a_i$ of $L_a$. By the definition of twin-sets, for every $z \in V(G) \setminus L_a$, either $X^a_i \subseteq N(z)$ or $X^a_i \cap N(z) = \emptyset$ holds. Since $V(G) \setminus L_v \subseteq V(G) \setminus L_a$, $X^a_i$ is also a twin-set of $L_v$, which implies $X^a_i$ is included in some twin-class $X^v_j$ of $L_v$. This argument indeed holds for twin-classes of $L_b$. Therefore, we have the above observation.

The intuition of our recurrence is as follows.
By Observation~\ref{obs:coarse}, every twin-class of $L_v$ can be obtained by merging some twin-classes of $L_a$ and of $L_b$.
This means that every $t_v$-legitimate cut of $G[L_v]$ for a valid tuple $t_v$ for $v$ can be obtained from some $t_a$-legitimate cut and $t_b$-legitimate cut for valid tuples for $a$ and $b$, respectively.
Moreover, for every pair of twin-classes $X^a_i$ of $L_a$ and $X^b_j$ of $L_b$, either there are no edges between them or every vertex in $X^a_i$ is adjacent to every vertex in $X^b_j$ as $X^a_i$ is a twin-set of $L_v$.
Therefore, the number of edges in the cutset of a cut $(S, L_v \setminus S)$ between $X^a_i$ and $X^b_j$ depends only on the cardinality of $X^a_i \cap S$ and $X^b_j \cap S$ rather than actual cuts $(S \cap X^a_i, (L_a \setminus S) \cap X^a_i)$ and $(S \cap X^b_i, (L_b \setminus S) \cap X^b_i)$.

Now, we formally describe this idea. Let $X^v$ be a twin-class of $L_v$.
We denote by $I_a(X^v)$ (resp. $I_b(X^v)$) the set of indices $i$ such that $X^a_i$ (resp. $X^b_i$) is included in $X^v$ and by $\CX^a(X^v)$ (resp. $\CX^b(X^v)$) the set $\{X^a_i : i \in I_a(X^v)\}$ (resp. $\{X^b_i : i \in I_b(X^v)\}$).
For $X^a \in \CX^a(X^v)$ and $X^b \in \CX^a(X^v)$, we say that $X^a$ is adjacent to $X^b$ if every vertex in $X^a$ is adjacent to every vertex in $X^b$ and otherwise $X^a$ is not adjacent to $X^b$. 
This adjacency relation naturally defines a bipartite graph whose vertex set is $\CX^a(X^v) \cup \CX^b(X^v)$.
We say that a subset of twin-classes of $\CX^a(X^v) \cup \CX^b(X^v)$ is {\em non-trivially connected} if it induces a connected bipartite graph with at least twin-classes.
Let $S \subseteq X^v$.
To make $G[S]$ (and $G[X^v \setminus S]$) connected, the following observation is useful.
\begin{observation}\label{obs:connect}
    Suppose $S \subseteq X^v$ has a non-empty intersections with at least two twin-classes of $\CX^a(X^v) \cup \CX^b(X^v)$.
    Then, $G[S]$ is connected if and only if the twin-classes having a non-empty intersection with $S$ are non-trivially connected.
\end{observation}
This observation immediately follows from the fact that every vertex in a twin-class is adjacent to every vertex in an adjacent twin-class and is not adjacent to every vertex in a non-adjacent twin-class.

Let $t_v = (p^v_1, \overline{p}^v_1, \ldots, p^v_{w_v}, \overline{p}^v_{w_v}, c^v_1, \overline{c}^v_2, \ldots, c^v_{w_v}, \overline{c}^v_{w_v})$ be a valid tuple for $v$.
For notational convenience, we use ${\bf p}^v$ to denote $(p^v_1, \overline{p}^v_1, \ldots, p^v_{w_v}, \overline{p}^v_{w_v})$ and ${\bf c}^v$ to denote $(c^v_1, \overline{c}^v_2, \ldots, c^v_{w_v}, \overline{c}^v_{w_v})$ for each node $v$ in $T$.
For valid tuples $t_a = ({\bf p}^a, {\bf c}^a)$ for $a$ and $t_b = ({\bf p}^b, {\bf c}^b)$ for $b$,
we say that {\em $t_v$ is consistent with the pair $(t_a, t_b)$} if for each $1\le i \le w_v$, 
\begin{itemize}
    \item[C1] $p^v_i = \sum_{j \in I_a(X^v_i)} p^a_j + \sum_{j \in I_b(X^v_i)} p^b_j$;
    \item[C2] $\overline{p}^v_i = \sum_{j \in I_a(X^v_i)} \overline{p}^a_j + \sum_{j \in I_b(X^v_i)} \overline{p}^b_j$;
    \item[C3] if $c^v_i = 1$, either (1) $\{X^a_j : j \in I_a(X^v), p^a_j > 0\} \cup \{X^b_j : j \in I_b(X^v), p^b_j > 0\}$ is non-trivially connected or (2) exactly one of $\{p^s_j : s \in \{a, b\}, 1 \le j \le w_s\}$ is positive, say $p^s_j$, and $c^s_j = 1$;
    \item[C4] if $\overline{c}^v_i = 1$, either (1) $\{X^a_j : j \in I_a(X^v), \overline{p}^a_j > 0\} \cup \{X^b_j : j \in I_b(X^v), \overline{p}^b_j > 0\}$ is non-trivially connected or (2) exactly one of $\{\overline{p}^s_j : s \in \{a, b\}, 1 \le j \le w_s\}$ is positive, say $\overline{p}^s_j$, and $\overline{c}^s_j = 1$.
\end{itemize}

\begin{lemma}\label{lem:cw:rec}
    \[
        \mc(v, t_v) = \max_{t_a, t_b} \left(\mc(a, t_a) + \mc(b, t_b) + \sum_{\substack{X^a_i \in \CX^a, X^b_j \in \CX^b\\
        X^a_i, X^b_j: \text{adjacent}}} (p^a_i\overline{p}^b_j + p^b_j\overline{p}^a_i)  \right),
    \]
    where the maximum is taken over all consistent pairs $(t_a, t_b)$.
\end{lemma}

\begin{proof}
We first show that the left-hand side is at most the right-hand side.
Suppose $(S, L_v \setminus S)$ be a $t_v$-legitimate cut of $G[L_v]$ whose size is equal to $\mc(v, t_v)$.
Let $S_a = S \cap L_a$ and $S_b = S \cap L_b$.
We claim that $(S_a, L_a \setminus S_a)$ is a $t_a$-legitimate cut of $G[L_a]$ for some valid tuple $t_a$ for $a$. 
This is obvious since we set $p^a_i = |S_a \cap X^a_i|$, $\overline{p}^a_i = |(L_a \setminus S_a) \cap X^a_i|$, $c^a_i = 1$ if $G[S_a \cap X^a_i]$ is connected, and $c^a_i = 1$ if $G[(L_a \setminus S_a) \cap X^a_i]$ is connected, which yields a valid tuple $t_a$ for $a$.
We also conclude that  $(S_b, L_b \setminus S_b)$ is a $t_b$-legitimate cut of $G[L_b]$ for some valid tuple $t_b$ for $b$.
Moreover, the number of cut edges between twin-class $X^a_i$ of $L_a$ and twin-class $X^b_j$ of $L_b$ is $|S_a \cap X^a_i|\cdot |(L_b\setminus S_b) \cap X^b_j| + |S_b \cap X^b_j|\cdot|(L_b \setminus S_a)\cap X^a_i| = p^a_i\overline{p}^b_j + p^b_j\overline{p}^a_i$ if $X^a_i$ and $X^b_j$ is adjacent, zero otherwise.
Therefore, the left-hand side is at most the right-hand side.

To show the converse direction, suppose $(S_a, L_a \setminus S_a)$ is a $t_a$-legitimate cut of $G[L_a]$ and $(S_b, L_b \setminus S_b)$ is a $t_b$-legitimate cut of $G[L_b]$, where $t_v$ is consistent with $(t_a, t_b)$ and the sizes of the cuts are $\mc(a, t_a)$ and $\mc(b, t_b)$, respectively.
We claim that $(S_a \cup S_b, L_v \setminus (S_a \cup S_b))$ is a $t_v$-legitimate cut of $G[L_v]$.
Since $t_v$ is consistent with $(t_a, t_b)$, for each $1 \le i \le w_v$, we have $p^v_i = \sum_{j \in I_a(X^v_i)} p^a_j + \sum_{j \in I_b(X^v_i)} p^b_j = \sum_{1 \le j \le w_a}|S_a \cap X^i_v| + \sum_{1 \le j \le w_b}|S_b \cap X^i_v| = |(S_a \cup S_b) \cap X^i_v|$. Symmetrically, we have $\overline{p}^i = |(L_v \setminus (S_a \cup S_b)) \cap X^v_i|$.
If $c^v_i = 1$, by condition C3 of the consistency, either (1) $\{X^a_j : j \in I_a(X^v), p^a_j > 0\} \cup \{X^b_j : j \in I_b(X^v), p^b_j > 0\}$ is non-trivially connected or (2) exactly one of $\{p^s_j : s \in \{a, b\}, 1 \le j \le w_s\}$ is positive, say $p^s_j$, and $c^s_j = 1$. 
If (1) holds, by Observation~\ref{obs:connect}, $G[(S_a \cap S_b) \cap X^i_v]$ is connected. Otherwise, as $c^s_j = 1$, $G[S_s \cap X^i_v] = G[(S_a \cup S_b) \cap X^v_i]$ is also connected.
By a symmetric argument, we conclude that $G[(L_v \setminus (S_a \cup S_b)) \cap X^i_v]$ is connected if $\overline{c}^v_i = 1$. Therefore the cut $(S_a \cup S_b, L_v \setminus (S_a \cup S_b))$ is $t_v$-legitimate.
Since the cut edges between two twin-classes of $L_a$ is counted by $\mc(a, t_a)$ and those between two twin-classes of $L_v$ is counted by $\mc(b, t_b)$. 
Similar to the forward direction, the number of cut edges between a twin-class of $L_a$ and a twin-class of $L_b$ can be counted by the third term in the right-hand side of the equality.
Hence, the left-hand side is at least right-hand side.
\end{proof}

\begin{theorem}\label{thm:cw:dp}
    \textsc{Connected Maximum Cut} and \textsc{Maximum Minimal Cut} can be computed in time $n^{O(w)}$ provided that a $w$-expression tree of $G$ is given as input.
\end{theorem}

\begin{proof}
    From a $w$-expression tree of $G$, we can obtain a decomposition tree $(T, \phi)$ of width at most $w$ in $O(n^2)$ time using Rao's algorithm \cite{Rao2008}.
    Based on this decomposition, we evaluate the recurrence in Lemma~\ref{lem:cw:rec} in a bottom-up manner. The number of valid tuples for each node of $T$ is at most $4^wn^w$. For each internal node $v$ and for each valid tuple $t_v$ for $v$, we can compute $\mc(v, t_v)$ in $(4^wn^w)^2n^{O(1)}$ time. Overall, the running time of our algorithm is $n^{O(w)}$.
    Let $r$ be the root of $T$.
    For \textsc{Connected Maximum Cut}, by the definition of legitimate cuts, we should take the maximum value among $\mc(r, (i, n-i, 1, j))$ for $1 \le i < n$ and $j \in \{0, 1\}$. Note that as $L_v$ has only one twin-class, the length of valid tuples is exactly four. For \textsc{Maximum Minimal Cut}, we should take the maximum value among $\mc(r, (i, n-i, 1, 1))$ for $1 \le i < n$.
\end{proof}

Since there is an algorithm that, given a graph $G$ and an integer $k$, either conclude that the clique-width of $G$ is more than $k$ or find a $(2^{k-1}-1)$-expression tree of $G$ in time $O(n^3)$ \cite{Hlineny2008,Oum2006,Oum2008}, {\sc Maximum Minimal Cut} and {\sc Connected Maximum Cut} are XP parameterized by the clique-width of the input graph.

\subsection{Twin-cover}\label{sec:twincover}
\textsc{Maximum Cut} is FPT when parameterized by  twin-cover number \cite{Ganian2015}.
In this section, we show that \textsc{Connected Maximum Cut} and \textsc{Maximum Minimal Cut} are also FPT when parameterized by  twin-cover number.

\begin{theorem}\label{thm:twincover:concut}
\textsc{Connected Maximum Cut} can be solved in time $O^*(2^{2^{\tc}+\tc})$.
\end{theorem}
\begin{proof}
We first compute a minimum twin-cover $X$ of $G = (V, E)$ in time $O^*(1.2738^{\tc})$~\cite{Ganian2015}.
Now, we have a twin-cover $X$ of size $\tc$. Recall that $G[V\setminus X]$ consists of vertex disjoint cliques and for each $u, v \in Z$ in a clique $Z$ of $G[V \setminus X]$, $N(u) \cap X = N(v) \cap X$.

We iterate over all possible subsets $X'$ of $X$ and compute the size of a  maximum cut $(S, V\setminus S)$ of $G$ with $S \cap X = X'$.

If $X' = \emptyset$, exactly one of the cliques of $G[V \setminus X]$ intersects $S$ as $G[S]$ is connected.
Thus, we can compute a maximum cut by finding a maximum cut for each clique of $G[V \setminus X]$, which can be done in polynomial time.

Suppose otherwise that $X' \neq \emptyset$.
We define a {\em type} of each clique $Z$ of $G[V \setminus X]$.
The type of $Z$, denoted by $T(Z)$, is $N(Z) \cap X$. Note that there are at most $2^{\tc}-1$ types of cliques in $G[V \setminus X]$.

For each type of cliques, we guess that $S$ has an intersection with this type of cliques. There are at most $2^{2^\tc - 1}$ possible combinations of types of cliques.
Let $\mathcal T$ be the set of types in $G[V \setminus X]$. 
For each guess $\mathcal T' \subseteq \mathcal T$, we try to find a maximum cut $(S, V \setminus S)$ such that $G[S]$ is connected, $S \cap X = X'$, for each $T \in \mathcal T'$, at least one of the cliques of type $T$ has an intersection with $S$, and for each $T \notin \mathcal T'$, every clique of type $T$ has no intersection with $S$.
We can easily check if $G[S]$ will be connected as $S$ contains a vertex of a clique of type $T \in \mathcal T'$.
Consider a clique $Z$ of type $T(Z) = X'' \subseteq X$. 
Since every vertex in $Z$ has the same neighborhood in $X$, we can determine the number of cut edges incident to $Z$ from the cardinality of $S \cap Z$.
More specifically, if $|S \cap Z| = p$, the number of cut edges incident to $Z$ is equal to $p(|Z|-p)+p|X''\cap (X \setminus X')| + (|Z|-p)|X'' \cap X'|$.
Moreover, we can independently maximize the number of cut edges incident to $Z$ for each clique $Z$ of $G[V \setminus X]$.

Overall, for each $X' \subseteq X$ and for each set of types $\mathcal T'$, we can compute a maximum connected cut with respect to $X'$ and $\mathcal T'$ in polynomial time.
Therefore, the total running time is bounded by $O^*(2^{2^\tc+\tc})$.
\end{proof}

\begin{theorem}\label{thm:twincover:maxmin}
\textsc{Maximum Minimal Cut} can be solved in time $O^*(2^{\tc}3^{2^{\tc}})$.
\end{theorem}
\begin{proof}
We design an $O^*(2^{\tc}3^{2^{\tc}})$-time algorithm for \textsc{Maximum Minimal Cut}, where $\tc$ is the size of a minimum twin-cover of $G = (V, E)$.
This is quite similar to the one for \textsc{Connected Maximum Cut} developed in Section\ref{sec:twincover}.
As with an algorithm for \textsc{Connected Maximum Cut}, we first compute a minimum twin-cover $X$ in time $O^*(1.2738^{\tc})$~\cite{Ganian2015}.
Then we guess all $2^{\tc}$ possible subsets $X'\subseteq X$ and compute the size of maximum cut $(S,V\setminus S)$ of $G$ with $S\cap X=X'$.

If $X'=\emptyset$, exactly one of the cliques of $G[V\setminus X]$ intersects $S$ due to the connectivity of $G[S]$. Thus, we can compute a maximum cut in polynomial time.
Note that $G[V\setminus S]$ is also connected because $X\subseteq  V\setminus S$.
Similarly, when $S\cap X= \emptyset$, we can compute a maximum cut in polynomial time.
We are also done for the case where $X' = X$ by a symmetric argument.
Thus, in the following, we assume that our guess $X'$ is non-empty and proper subset of $X$.

For each guess $X' \subseteq X$, we further guess each type of cliques in $G[V \setminus X]$ has an intersection with only $S$, with only $V \setminus S$, or with both $S$ and $V \setminus S$
For each guess, we can easily check $S$ and $V \setminus S$ will be connected and maximize the size of a cut in polynomial time as in Section~\ref{sec:twincover}.
Since there are at most $2^{\tc}$ types of cliques in $G[V \setminus X]$, the total running time is $O^*(2^{\tc}3^{2^{\tc}})$.
\end{proof}

\subsection{Solution size}\label{sec:FPT:solutionsize}
In this section, we give FPT algorithms parameterized by the solution size for \textsc{Connected Maximum Cut} and  \textsc{Maximum Minimal Cut}. 
To show this, we use the following theorem.
\revised{
\begin{theorem}[\cite{Birmele2007}]\label{EGT}
The Cartesian product $C_ k\times K_2$ of a $k$-circuit with $K_2$ is called a \emph{$k$-prism}.
If $G$ contains no $k$-prism as a minor, $\tw(G)=O(k^2)$.
\end{theorem}
Then we have the following theorem.
\begin{theorem}\label{FPT:solution}
\textsc{Connected Maximum Cut} and \textsc{Maximum Minimal Cut} can be solved in time $O^*(2^{O(k^2)})$ where $k$ is the solution size.
\end{theorem}
\begin{proof}
We first determine whether the tree-width of $G$ is $O(k^2)$ in time $O^*(2^{O(k)})$ by using the algorithm in \cite{bod2016}. 
If $\tw(G)=O(k^2)$, the algorithm in \cite{bod2016} outputs a tree decomposition of width $O(k^2)$. Thus, we apply the dynamic programming algorithms based on tree decompositions described in Section~\ref{sec:treewidth}, and the running time is $O^*(2^{O(k^2)})$.
Otherwise, we can conclude that $G$ has a minimal cut (and also a connected cut) of size at least $k$.
To see this, consider a $k$-prism minor of $G$. Then, we take $k$ ``middle edges'' corresponding to $K_2$ in the $k$-prism minor and add some edges to make these edges form a cutset of some minimal cut of $G$. The size of such a cut is at least $k$ and hence $G$ has a minimal cut and a connected cut of size at least $k$.
\end{proof}
}

For \textsc{Connected Maximum Cut}, we can further improve the running time by giving an $O^*(9^k)$-time algorithm.

In \cite{Fellows2009}, Fellows et al. proposed a ``Win/Win'' algorithm that outputs in linear time either a spanning tree of $G$ having at least $k$ leaves, or a path decomposition of $G$ of width at most $2k$. If $G$ has such a spanning tree, we can construct a cut $(S, V \setminus S)$ of size at least $k$ by taking the internal vertices of the tree for $S$.
Clearly, $G[S]$ is connected, and hence we are done in this case. 
Otherwise, we have a path decomposition of width at most $2k$. Thus, we can compute \textsc{Connected Maximum Cut} on such a path decomposition by using an $O^*(3^{\tw})$-algorithm in Section \ref{sec:treewidth}. 
\begin{theorem}\label{thm:FPTk:Connected}
There is a Monte-Carlo algorithm that solves \textsc{Connected Maximum Cut}  in time $O^*(9^{k})$. It cannot give false positives
and may give false negatives with probability at most 1/2.
\end{theorem}
Also, using the rank-based algorithm in Theorem~\ref{thm:rank-based:concut}, we obtain an $O^*(38.2^k)$-time deterministic algorithm for \textsc{Connected Maximum Cut}. Note that our rank-based algorithm in Theorem~\ref{thm:rank-based:concut} runs in time $O^*((1+2^{\omega})^{\pw})$ on a path decomposition and $(1+2^{\omega})^2< 38.2$, where $\omega<2.3727$ is the exponent of matrix multiplication.
\begin{theorem}\label{thm:FPTk:Connected:rank-based}
There is an $O^*(38.2^k)$-time deterministic algorithm for \textsc{Connected Maximum Cut}.
\end{theorem}

As for kernelization, it is not hard to see that \textsc{Connected Maximum Cut} and \textsc{Maximum Minimal Cut} do not admit a polynomial kernelization unless NP~$\subseteq~$coNP$/$poly since both problems are trivially OR-compositional \cite{Bodlaender2009}; at least one of graphs $G_1, G_2, \ldots G_t$ have a connected/minimal cut of size at least $k$ if and only if their disjoint union $G_1 \cup G_2 \cup \cdots \cup G_t$ has.

\begin{theorem}\label{thm:NoPolyKernel}
Unless ${\rm NP}~\subseteq {\rm coNP}/{\rm poly}$, \textsc{Maximum Minimal Cut} and \textsc{Connected Maximum Cut} admit no polynomial kernel parameterized by the solution size.
\end{theorem}

\section{Conclusion and Remark}
In this paper, we studied two variants of \textsc{Max Cut}, called \textsc{Connected Maximum Cut} and \textsc{Maximum Minimal Cut}. 
We showed that both problems are NP-complete  even on planar bipartite graphs and split graphs. 
For the parameterized complexity, we gave FPT algorithms parameterized by tree-width, twin-cover number, and the solution size, respectively.
Moreover, we designed XP-algorithms parameterized by clique-width.


Finally, we mention our problems on weighted graphs.
It is not hard to see that \textsc{Connected Maximum Cut} and \textsc{Maximum Minimal Cut} remain to be FPT with respect to tree-width.
However, our results with respect to clique-width and twin-cover number would not be extended to weighted graphs since both problems are NP-hard on $0$-$1$ edge-weighted complete graphs.



\bibliography{ref}

\begin{thebibliography}{10}

\bibitem{Bazgan2018}
C.~Bazgan, L.~Brankovic, K.~Casel, H.~Fernau, K.~Jansen, K.-M. Klein,
  M.~Lampis, M.~Liedloff, J.~Monnot, and V.~T. Paschos.
\newblock The many facets of upper domination.
\newblock {\em Theoretical Computer Science}, 717:2--25, 2018.

\bibitem{Birmele2007}
E.~Birmel{\'e}, J.~A. Bondy, and B.~A. Reed.
\newblock {\em Brambles, Prisms and Grids}, pages 37--44.
\newblock Birkh{\"a}user Basel, Basel, 2007.

\bibitem{Bodlaender2015}
H.~L. Bodlaender, M.~Cygan, S.~Kratsch, and J.~Nederlof.
\newblock Deterministic single exponential time algorithms for connectivity
  problems parameterized by treewidth.
\newblock {\em Information and Computation}, 243:86--111, 2015.

\bibitem{Bodlaender2009}
H.~L. Bodlaender, R.~G. Downey, M.~R. Fellows, and D.~Hermelin.
\newblock On problems without polynomial kernels.
\newblock {\em Journal of Computer and System Sciences}, 75(8):423--434, 2009.

\bibitem{bod2016}
H.~L. Bodlaender, P.~G. Drange, M.~S. Dregi, F.~V. Fomin, D.~Lokshtanov, and
  M.~{Pilipczuk}.
\newblock A $c^k n$ 5-approximation algorithm for treewidth.
\newblock {\em SIAM Journal on Computing}, 45(2):317--378, 2016.

\bibitem{Bodlaender1995b}
H.~L. Bodlaender, J.~R. Gilbert, H.~Hafsteinsson, and T.~Kloks.
\newblock Approximating treewidth, pathwidth, frontsize, and shortest
  elimination tree.
\newblock {\em Journal of Algorithms}, 18(2):238--255, 1995.

\bibitem{bodlaender2000}
H.~L. Bodlaender and K.~Jansen.
\newblock On the complexity of the maximum cut problem.
\newblock {\em Nordic Journal of Computing}, 7(1):14--31, 2000.

\bibitem{Boria2015}
N.~Boria, F.~D. Croce, and V.~T. Paschos.
\newblock On the max min vertex cover problem.
\newblock {\em Discrete Applied Mathematics}, 196:62--71, 2015.

\bibitem{Boyaci2017}
A.~Boyacı, T.~Ekim, and M.~Shalom.
\newblock A polynomial-time algorithm for the maximum cardinality cut problem
  in proper interval graphs.
\newblock {\em Information Processing Letters}, 121:29--33, 2017.

\bibitem{Bui-Xuan2013}
B.-M. Bui-Xuan, O.~Suchý, J.~A. Telle, and M.~Vatshelle.
\newblock Feedback vertex set on graphs of low clique-width.
\newblock {\em European Journal of Combinatorics}, 34(3):666--679, 2013.

\bibitem{Carvajal2013}
R.~Carvajal, M.~Constantino, M.~Goycoolea, J.~P. Vielma, and A.~Weintraub.
\newblock Imposing connectivity constraints in forest planning models.
\newblock {\em Operations Research}, 61(4):824--836, 2013.

\bibitem{Chaourar2017}
B.~Chaourar.
\newblock A linear time algorithm for a variant of the {MAX} {CUT} problem in
  series parallel graphs.
\newblock {\em Advances in Operations Research}, pages 1267108:1--1267108:4,
  2017.

\bibitem{Chaourar2019}
B.~Chaourar.
\newblock Connected max cut is polynomial for graphs without ${K}_5\setminus e$
  as a minor.
\newblock {\em CoRR}, abs/1903.12641, 2019.

\bibitem{Courcelle2000}
B.~Courcelle and S.~Olariu.
\newblock Upper bounds to the clique width of graphs.
\newblock {\em Discrete Applied Mathematics}, 101(1):77--114, 2000.

\bibitem{Cygan2012}
M.~Cygan.
\newblock Deterministic parameterized connected vertex cover.
\newblock In {\em {SWAT} 2012}, pages 95--106, 2012.

\bibitem{Cygan2015}
M.~Cygan, F.~V. Fomin, {\L}.~Kowalik, D.~Lokshtanov, D.~Marx, M.~Pilipczuk,
  M.~Pilipczuk, and S.~Saurabh.
\newblock {\em Parameterized Algorithms}.
\newblock Springer International Publishing, 2015.

\bibitem{Cygan2011}
M.~Cygan, J.~Nederlof, M.~Pilipczuk, M.~Pilipczuk, J.~M. M.~van Rooij, and
  J.~O. Wojtaszczyk.
\newblock Solving connectivity problems parameterized by treewidth in single
  exponential time.
\newblock In {\em {FOCS} 2011}, pages 150--159, 2011.

\bibitem{Berg2009}
M.~de~Berg and A.~Khosravi.
\newblock Finding perfect auto-partitions is {NP}-hard.
\newblock In {\em {EuroCG} 2009}, pages 255--258., 2009.

\bibitem{Demange1999}
M.~Demange.
\newblock A note on the approximation of a minimum-weight maximal independent
  set.
\newblock {\em Computational Optimization and Applications}, 14(1):157--169,
  1999.

\bibitem{Diestel2012}
R.~Diestel.
\newblock {\em Graph Theory, 4th Edition}, volume 173 of {\em Graduate texts in
  mathematics}.
\newblock Springer, 2012.

\bibitem{Diaz2007}
J.~Díaz and M.~Kamiński.
\newblock Max-cut and max-bisection are {NP}-hard on unit disk graphs.
\newblock {\em Theoretical Computer Science}, 377(1):271--276, 2007.

\bibitem{Fellows2009}
M.~R. Fellows, D.~Lokshtanov, N.~Misra, M.~Mnich, F.~Rosamond, and S.~Saurabh.
\newblock The complexity ecology of parameters: An illustration using bounded
  max leaf number.
\newblock {\em Theory of Computing Systems}, 45(4):822--848, 2009.

\bibitem{Fomin2014}
F.~V. Fomin, P.~Golovach, D.~Lokshtanov, and S.~Saurabh.
\newblock Almost optimal lower bounds for problems parameterized by
  clique-width.
\newblock {\em SIAM Journal on Computing}, 43(5):1541--1563, 2014.

\bibitem{Ganian2015}
R.~Ganian.
\newblock {Improving Vertex Cover as a Graph Parameter}.
\newblock {\em {Discrete Mathematics and Theoretical Computer Science}},
  17(2):77--100, 2015.

\bibitem{GJ1979}
M.~R. Garey and D.~S. Johnson.
\newblock {\em Computers and Intractability: A Guide to the Theory of
  NP-Completeness}.
\newblock W. H. Freeman \& Co., New York, NY, USA, 1979.

\bibitem{Goemans1995}
M.~X. Goemans and D.~P. Williamson.
\newblock Improved approximation algorithms for maximum cut and satisfiability
  problems using semidefinite programming.
\newblock {\em Journal of the ACM}, 42(6):1115--1145, 1995.

\bibitem{Grimm2019}
V.~Grimm, T.~Kleinert, F.~Liers, M.~Schmidt, and G.~Zöttl.
\newblock Optimal price zones of electricity markets: a mixed-integer
  multilevel model and global solution approaches.
\newblock {\em Optimization Methods and Software}, 34(2):406--436, 2019.

\bibitem{Guha1998}
S.~Guha and S.~Khuller.
\newblock Approximation algorithms for connected dominating sets.
\newblock {\em Algorithmica}, 20(4):374--387, 1998.

\bibitem{Guruswami1999}
V.~Guruswami.
\newblock Maximum cut on line and total graphs.
\newblock {\em Discrete Applied Mathematics}, 92(2):217--221, 1999.

\bibitem{Hadlock1975}
F.~Hadlock.
\newblock Finding a maximum cut of a planar graph in polynomial time.
\newblock {\em SIAM Journal on Computing}, 4(3):221--225, 1975.

\bibitem{Haglin1991}
D.~J. {Haglin} and S.~M. {Venkatesan}.
\newblock Approximation and intractability results for the maximum cut problem
  and its variants.
\newblock {\em IEEE Transactions on Computers}, 40(1):110--113, 1991.

\bibitem{Hajiaghayi2015}
M.~T. Hajiaghayi, G.~Kortsarz, R.~MacDavid, M.~Purohit, and K.~Sarpatwar.
\newblock Approximation algorithms for connected maximum cut and related
  problems.
\newblock In {\em {ESA} 2015}, pages 693--704, 2015.

\bibitem{Hanaka2017}
T.~Hanaka, H.~L. Bodlaender, T.~C. van~der Zanden, and H.~Ono.
\newblock On the maximum weight minimal separator.
\newblock In {\em {TAMC} 2017}, pages 304--318, 2017.

\bibitem{Hlineny2008}
P.~Hliněný and S.~Oum.
\newblock Finding branch-decompositions and rank-decompositions.
\newblock {\em SIAM Journal on Computing}, 38(3):1012--1032, 2008.

\bibitem{Karp1972}
R.~M. Karp.
\newblock {\em Reducibility among Combinatorial Problems}, pages 85--103.
\newblock Springer US, Boston, MA, 1972.

\bibitem{Khoshkhah2019}
K.~Khoshkhah, M.~K. Ghadikolaei, J.~Monnot, and F.~Sikora.
\newblock Weighted upper edge cover: Complexity and approximability.
\newblock In {\em {WALCOM} 2019}, pages 235--247, 2019.

\bibitem{Mahajan1999}
M.~Mahajan and V.~Raman.
\newblock Parameterizing above guaranteed values: Maxsat and maxcut.
\newblock {\em Journal of Algorithms}, 31(2):335--354, 1999.

\bibitem{MVV1987}
K.~Mulmuley, U.~V. Vazirani, and V.V. Vazirani.
\newblock Matching is as easy as matrix inversion.
\newblock {\em Combinatorica}, 7(1):105--113, 1987.

\bibitem{Orlova1972}
G.~I. Orlova and Y.~G. Dorfman.
\newblock Finding the maximal cut in a graph.
\newblock {\em Engineering Cyvernetics}, 10(3):502--506, 1972.

\bibitem{Oum2008}
S.~Oum.
\newblock Approximating rank-width and clique-width quickly.
\newblock {\em ACM Transactions on Algorithms}, 5(1):10:1--10:20, 2008.

\bibitem{Oum2006}
S.~Oum and P.~Seymour.
\newblock Approximating clique-width and branch-width.
\newblock {\em Journal of Combinatorial Theory, Series B}, 96(4):514--528,
  2006.

\bibitem{Raman2007}
V.~Raman and S.~Saurabh.
\newblock Improved fixed parameter tractable algorithms for two “edge”
  problems: {MAXCUT} and {MAXDAG}.
\newblock {\em Information Processing Letters}, 104(2):65--72, 2007.

\bibitem{Rao2008}
M.~Rao.
\newblock Clique-width of graphs defined by one-vertex extensions.
\newblock {\em Discrete Mathematics}, 308(24):6157--6165, 2008.

\bibitem{Robertson1986}
N.~Robertson and P.~D. Seymour.
\newblock Graph minors. {V.} excluding a planar graph.
\newblock {\em Journal of Combinatorial Theory, Series B}, 41(1):92--114, 1986.

\bibitem{Saurabh2018}
S.~Saurabh and M.~Zehavi.
\newblock {Parameterized Complexity of Multi-Node Hubs}.
\newblock In {\em {IPEC} 2018}, volume 115, pages 8:1--8:14, 2019.

\bibitem{Vicente2008}
S.~{Vicente}, V.~{Kolmogorov}, and C.~{Rother}.
\newblock Graph cut based image segmentation with connectivity priors.
\newblock In {\em {CVPR} 2008}, pages 1--8, 2008.

\bibitem{Yannakakis1980}
M.~Yannakakis and F.~Gavril.
\newblock Edge dominating sets in graphs.
\newblock {\em SIAM Journal on Applied Mathematics}, 38(3):364--372, 1980.

\bibitem{Zehavi2017}
M.~Zehavi.
\newblock Maximum minimal vertex cover parameterized by vertex cover.
\newblock {\em SIAM Journal on Discrete Mathematics}, 31(4):2440--2456, 2017.

\end{thebibliography}

\end{document}